\documentclass[11pt]{article}

\usepackage[letterpaper, margin=0.95in]{geometry}

\usepackage[utf8]{inputenc}
\usepackage[T1]{fontenc}

\usepackage{amsmath}
\usepackage{amsfonts}
\usepackage{amsthm}

\usepackage{algorithm}
\usepackage{algorithmic}
\usepackage{graphicx}
\usepackage{textcomp}
\usepackage{xcolor}
\usepackage{mathrsfs}
\usepackage{complexity}
\usepackage{hyperref}
\usepackage[capitalize]{cleveref}
\usepackage{mathtools}
\usepackage[group-separator = {,}, group-minimum-digits=4]{siunitx}

\crefname{section}{Section}{Sections}
\crefname{figure}{Figure}{Figures}
\crefname{equation}{Equation}{Equations}
\crefname{property}{Property}{Properties}
\crefname{algorithm}{Algorithm}{Algorithms}

\def\BibTeX{{\rm B\kern-.05em{\sc i\kern-.025em b}\kern-.08em
    T\kern-.1667em\lower.7ex\hbox{E}\kern-.125emX}}

\newcommand{\model}{\mathcal{M}}

\newcommand{\features}{\mathcal{F}}
\newcommand{\clauses}{\mathcal{D}}
\newcommand{\clause}{\mathcal{D}}

\newcommand{\literals}{\mathcal{L}}
\newcommand{\literal}{\ell}
\newcommand{\interaction}{I}
\newcommand{\validInteractions}{\mathcal{I}}
\newcommand{\sample}{S}

\newcommand{\exclusiveInteractions}{E}
\newcommand{\alg}[1]{\text{\textsc{#1}}}
\newcommand{\lbValue}{LB}
\newcommand{\lbIP}{\alg{OptLB}}
\newcommand{\lbIS}{\alg{LBSearch}}
\newcommand{\ie}[1]{(i.e., #1)}
\newcommand{\eg}[1]{(e.g. #1)}
\newcommand{\cf}[1]{(cf. #1)}
\newcommand{\omitted}[1]{}

\newtheorem{property}{Property}
\newtheorem{definition}{Definition}
\newtheorem{lemma}{Lemma}
\newtheorem{theorem}{Theorem}

\title{How Low Can We Go? \\ Minimizing Interaction Samples for Configurable Systems}

\author{
Dominik Krupke\\
TU Braunschweig, Germany\\
\texttt{krupke@ibr.cs.tu-bs.de}
\and
Ahmad Moradi\\
TU Braunschweig, Germany\\
\texttt{moradi@ibr.cs.tu-bs.de}
\and
Michael Perk\\
TU Braunschweig, Germany\\
\texttt{perk@ibr.cs.tu-bs.de}
\and
Phillip Keldenich\\
TU Braunschweig, Germany\\
\texttt{keldenich@ibr.cs.tu-bs.de}
\and
Gabriel Gehrke\\
TU Braunschweig, Germany\\
\texttt{ggehrke@ibr.cs.tu-bs.de}
\and
Sebastian Krieter\\
TU Braunschweig, Germany\\
\texttt{sebastian.krieter@tu-braunschweig.de}
\and
Thomas Thüm\\
TU Braunschweig, Germany\\
\texttt{thomas.thuem@tu-braunschweig.de}
\and
Sándor P. Fekete\\
TU Braunschweig, Germany\\
\texttt{s.fekete@tu-bs.de}
}

\date{}

\begin{document}

\maketitle
\begin{abstract}
Modern software systems are typically configurable, a fundamental prerequisite for wide applicability and reusability. 
This flexibility poses an extraordinary challenge for quality assurance, as the enormous number of possible configurations makes it impractical to test each of them separately.
This is where \emph{t-wise interaction sampling} can be used to systematically cover the configuration space and detect unknown feature interactions. 
Over the last two decades, numerous algorithms for computing small interaction samples have been studied, providing improvements for a range of heuristic results; nevertheless, it has remained unclear how much these results can still be improved.

We present a significant breakthrough: a fundamental framework,
based on the mathematical principle of \emph{duality},
for combining near-optimal solutions with provable lower bounds on 
the required sample size. This implies that we no longer need
to work on heuristics with marginal or no improvement, but 
can certify the solution quality by establishing a limit on the remaining gap;
in many cases, we can even prove optimality of achieved solutions.
This theoretical contribution also provides extensive practical improvements:
Our algorithm SampLNS was tested on 
\num[text-series-to-math=true]{47} small and medium-sized configurable systems from the
existing literature.
SampLNS can reliably find samples of smaller size than previous methods
in \SI[text-series-to-math=true]{85}{\percent} of the cases;
moreover, we can achieve and prove optimality of solutions for \SI[text-series-to-math=true]{63}{\percent} of all instances.
This makes it possible to avoid cumbersome efforts of
minimizing samples by researchers as well as practitioners, and substantially
save testing resources for most configurable systems.
\end{abstract}

\section{Introduction}

Modern software systems often provide hundreds or thousands of configuration
options \ie{\textit{features}} that enable customization of these systems to
specific user requirements~\cite{HZS+:EMSE16,QPV+:SoSyM17,sundermann2023evaluating}.
A challenge for testing these highly configurable
systems are potential interactions between features, which can lead to
unforeseen behavior and faults in configurations containing a particular
feature combination
\cite{BDC+:SETSS89,CKMR03,ABKS13,AMS+:TOSEM18}.
Due to the enormous number of potential feature interactions in
configurable systems, detecting such faults is challenging~\cite{AMS+:TOSEM18}.  A popular solution for finding faults in a
configurable system is product-based testing in combination with sampling
methods to create a representative list of configurations \ie{a
\textit{sample}}, which are then tested individually~\cite{RAK+:VaMoS13,MMCA:IST14,VAT+:SPLC18}.
The number of configurations, known as the sample size,
is a common proxy for approximating the testing effort.
Thus, minimizing the sample size can reduce the overall testing effort~\cite{HLHE:VaMoS13}.
However, to ensure reliability,
it is typically necessary to test each interaction between up to $t$ features~\cite{OGB:SPLC19,FVDF:JSS21}.
This yields the {\sc t-wise Interaction Sampling Problem}
($t$-ISP):
find a {sample} consisting of a 
minimum number of configurations, such that every valid combination of $t$ 
or less features is part of at least one configuration, where $t$ is an arbitrary but fixed positive integer.

In addition to its practical importance for modern software systems, Interaction Sampling
is also highly challenging from a theoretical point of view, as it contains the classical problem 
\textsc{Set Cover} as a special case: Given a universe $\mathcal{U}$ of elements $\{u_1,\ldots,u_n\}$
and a family $\mathcal{S}=\{S_1,\ldots,S_m\}$ of subsets 
of $\mathcal{U}$, each with a nonnegative cost 
$c(S_j)$, find a subfamily $\mathcal{C}:=\{S_j\mid j\in J\}$ with $J\subset \{1,\ldots,m\}$,
such that $\bigcup_{j\in J} S_j=\mathcal{U}$ and $\sum_{j\in J} c(S_j)$ is minimal.
\textsc{Set Cover} is an \NP-hard problem~\cite{karp1972reducibility}, so we cannot
expect to develop a deterministic algorithm that can guarantee to compute a provably
optimal solution in a worst-case runtime whose asymptotic growth is bounded by a polynomial
in the input size. This also implies that not only may it be difficult to \emph{find}
a good solution, it is also hard to \emph{certify} that an identified solution is indeed optimal,
because the logical complement (prove that there is \emph{no} solution that is better)
is co\NP-hard.

This combination of practical importance with theoretical difficulty of $t$-ISP
makes it plausible to resort to mere heuristics, such as greedy or evolutionary
approaches, which aim at finding decent solutions by
local improvement. (Even though some algorithms may be able to generate 
optimal samples~\cite{yamada2015optimization} for small instances
by solving a large incremental SAT formulation, their usefulness
is limited by their lack of scalability.) 
While such heuristics may yield decent samples~\cite{VAT+:SPLC18},
they cannot be expected to reliably compute optimal solutions; what is more, they
do not provide any estimate on solution quality, so there is no way of knowing how much
potential for further improvement exists - or even to certify whether an obtained
solution happens to be optimal. 
As a consequence, 
the sample size achieved by a new algorithm for $t$-ISP
can only be compared to other existing algorithms~\cite{MKR+:ICSE16}.
However, this means the achieved sample size is compared only against an
\textit{upper bound} \ie{the smallest known sample size for which there exists
a sample that achieves full $t$-wise coverage}, but not to a lower
bound or even an optimal sample. 
Currently, there exist dozens of different $t$-wise
sampling algorithms~\cite{VAT+:SPLC18,AZAB17,LFRE:IWCT15}, for which the
absolute performance 
is unknown, effectively using a lower bound of zero. 
In the absence of better lower bounds, it is completely unclear whether investigating 
performance improvements of these algorithms is a promising or a futile endeavor.

This is where Mathematical Optimization and Algorithm Engineering 
offer mathematically rigorous and practically useful methods. Over the years,
these areas have developed more refined ways to deal with the practical difficulties of 
finding good solutions for relevant instances of hard optimization problems.
One success story is offered by the well-known \textsc{Traveling Salesman Problem} of 
finding a shortest roundtrip through a given number of locations,
where such methods have made it possible to compute \emph{provably optimal} solutions
to instances with as many as 13,509 cities~\cite{applegate1998solution},
which amounts to both \emph{finding} a best of 13,509! $\approx 10^{50,000}$ possible
roundtrips and also \emph{proving} that none of the others is better; 
see~\cite{cook2011traveling} for full details. 

A crucial concept in this endeavor
is the notion of \emph{duality}. For a given ``primal'' minimization problem $\mathcal{A}$
with a cost function $a$,
the idea is to identify a suitable ``dual'' maximization problem $\mathcal{B}$ with a cost
function $b$, such that for any (not necessarily optimal) 
feasible solution $x$ of $\mathcal{A}$ and any 
feasible solution $y$ of $\mathcal{B}$, the objective values $a(x)$ of $x$ and
$b(y)$ of $y$ satisfy $a(x)\geq b(y)$. 
Thus, once a suitable problem $\mathcal{B}$ has been identified and duality has been established, 
\emph{any} feasible solution $y$ of $\mathcal{B}$ provides
a lower bound for the optimal value of $\mathcal{A}$, independent of how it was obtained. This
key observation allows it to combine a variety of approaches (including local
search for feasible solutions of $\mathcal{A}$ and $\mathcal{B}$) with mathematically rigorous
quality guarantees of the obtained feasible solutions. In particular, if 
we manage to find feasible solutions $x^*$ for $\mathcal{A}$ and $y^*$ for $\mathcal{B}$,
such that $a(x^*)=b(y^*)$, then any feasible solution $x$ of  $\mathcal{A}$ must
satisfy $a(x)\geq b(y^*)=a(x^*)$, so $y^*$ provides a certificate for the optimality
of $x^*$ for $\mathcal{A}$; this holds regardless of whether $x^*$ and $y^*$
were obtained by a time-consuming, systematic search or by a lucky guess.
Even if establishing optimality by matching dual solutions turns out
to be time consuming or elusive, we may be able to quickly establish a 
limited gap to an optimal solution; this is particularly useful for difficult
optimization problems for which finding \emph{some} good solution is
usually achieved much faster than proving that \emph{no} better solution exists.
(See \cref{subsec:duality} for a slightly more detailed technical introduction.)

\paragraph{Our Contributions}

We present a powerful algorithm for computing upper and lower bounds
for the problem of minimizing the sample size in software configuration problems such as
the $t$-ISP, along with a mathematically rigorous framework for
providing performance metrics.
While achieving tight lower or upper bounds (and thus, optimality)
cannot be guaranteed in general,
the practical applicability of our method is underscored through the provision of optimal or nearly optimal samples across a broad range of benchmark instances,
accompanied by quality certificates.
Notably, we were able to obtain a provably minimal sample size for the feature model \texttt{EMBToolkit} with \num{1179} features and \num{5414} clauses.
In detail, we provide the following.
\begin{itemize}
    \item We present a strategy for computing provable lower bounds for the optimum sample size in the $t$-ISP,
based on establishing a suitable problem and a mathematical proof of duality.
    \item We introduce an algorithm that minimizes the sample size with full-coverage guarantees and provides a quality certificate, i.e., lower bound.
    \item We demonstrate the practical usefulness of our approach by a comparison of \num{47} small to medium-sized configurable systems.
    The outcome are substantial improvements (on \SI{60}{\percent} of the instances we have an improvements of at least \SI{20}{\percent}), and tight performance bounds for \SI{55}{\percent} of the instances. 
    \item We provide all source code and data for reproduction and verification of our results.
\end{itemize}

\section{Preliminaries}

In this section, we formally define feature models and the $t$-wise interaction sampling problem.
We also introduce key algorithmic techniques that will be employed later in our approach,
including Mixed Integer Programming and Constraint Programming.
While these techniques can have exponential runtimes,
making them impractical for certain problem sizes, 
they remain powerful tools for solving many otherwise intractable problems.
Additionally, we present Large Neighborhood Search as a practical heuristic for scaling these techniques,
and discuss the concepts of duality and certificates,
which can be used to establish quality guarantees even when applying heuristic methods like LNS.

\subsection{Validity of Configurations}

The dependencies in a reconfigurable system can be expressed by
defining the set of valid configurations as those that satisfy a Boolean
formula.
To this end,
a \emph{feature model} $\model=\left( \features, \clauses \right)$ consists of
$n$ features $\features=\{1,2,\ldots, n\}$ and $m$ dependencies
$\clauses=\{\clause_1, \clause_2, \ldots, \clause_m\}$.  We denote the literals
of $\model$ by $\literals=\{-n$, $\ldots$, $-1,1$,$\ldots n\}$, where $i\in
\literals$ represents the \emph{activated} feature $|i|$, if $i>0$ and the 
\emph{deactivated}
feature otherwise.  A configuration can be expressed as a set of \emph{literals}
$C\subset \literals$, such that for every feature $i\in \features$ either $i\in
C$ or $-i \in C$.  For simplicity, let the dependencies $\clauses$ be given in
conjunctive normal form, for which each clause $\clause_i\in \clauses$ is a subset
of literals $\literals$.
Every Boolean formula can be expressed in this form, possibly with 
the help of auxiliary features~\cite{KKS+:ASE22}.
We denote the set of all valid configurations by $||\model{}||$.
The validity of a configuration $C$ is determined by $\forall \clause_i \in \clauses:
\clause_i\cap C\not= \emptyset$.
Thus, a valid configuration has to contain a literal of each clause
$\clause_i\in  \clauses$.  Note that the definition based on the conjunctive
normal form is purely for readability; our actual implementation supports
the direct use of more advanced constraints.

\subsection{\texorpdfstring{$t$}{t}-Wise Interaction Samples}

In $t$-wise interaction sampling, an \emph{interaction} $\interaction$ is
represented by a subset $I\subset \literals$ of $t$ literals; it is considered
\emph{valid} iff there exists a valid configuration $C\in ||\model{}||$ with $I\subseteq C$.
Correspondingly, an interaction $I$ is considered \emph{invalid}, if no valid
configuration $C$ with $I\subseteq C$ exists and thus, it can never appear in any
instantiation of the system.  For a set of valid interactions
$\validInteractions$, a \emph{sample} $\sample=\{C_1, C_2, \ldots\}\subseteq ||\model{}||$ is a set
of valid configurations such that $\forall I\in \validInteractions: \exists
C_i\in \sample: I\subset C_i$, i.e., every interaction of $\validInteractions$
is contained in at least one of the configurations.
We especially consider \emph{complete pairwise interaction sampling}, for which $t=2$
and $\validInteractions$ is the set of all valid interactions on the
concrete features of $\features$, i.e., the features that are not added as
composite or auxiliary features~\cite{TKES:SPLC11}.  The underlying
optimization problem $t$-ISP is to find a sample $\sample$ for $\validInteractions$ 
with a minimum number of configurations.

\subsection{Algorithmic Techniques}

We employ a number of powerful algorithmic techniques and implementations.
Fundamentally important are \emph{SAT solvers} for finding feasible assignments
or proving unsatisfiability for propositional formulas in conjunctive normal form.
SAT solvers are typically based on an algorithm called Conflict-Driven Clause
Learning~\cite{marques1999grasp}, which learns new clauses that must also be
satisfied whenever all clauses of the original formula are.  Modern
implementations, such as
Kissat~\cite{BiereFazekasFleuryHeisinger-SAT-Competition-2020-solvers}, can
routinely solve instances with hundreds of thousands of variables in acceptable
time.  Such a solver is also utilized by other sampling algorithms, such as ICPL~\cite{MHF:SPLC12} or
YASA~\cite{KTS+:VaMoS20}, and is implicitly used by us via \emph{CP-SAT}.

Also relevant are \emph{Constraint Programming} (CP)  solvers, such as CP-SAT 
(a part of Google's ortools~\cite{ortools}), which facilitate not only 
considering more complex constraints and variables than a SAT solver, 
but also optimizing an objective
function; this enables us to find not only feasible but actually optimal
assignments. By using a portfolio of techniques, especially \emph{lazy clause
generation}~\cite{feydy2009lazy}, CP-SAT can solve
many optimization models with hundreds of thousands of variables.  We show in
\cref{sec:upper-bounds:cpsat} how to use this tool to compute samples of
minimum size for small feature models.

A further highly effective technique to compute optimal solutions for
combinatorial optimization problems is \emph{Mixed Integer Programming} (MIP).
For a given set of linear constraints, such as $\sum_{j=1}^n a_{ij}x_j
\leq b_i$, the goal is to minimize or maximize a linear objective function $\sum_{i=1}^n c_ix_i$ 
over a set of real or integer variables $x_i$.  MIP is
a well-known \NP-hard problem, but commercial MIP solvers (such as IBM~ILOG~CPLEX or
Gurobi~\cite{gurobi}) can routinely solve instances of small to medium size in
acceptable time to (near-) optimality.  These solvers are based on the
branch-and-bound algorithm template, refined with a large portfolio of
techniques, including the generation of additional, valid \emph{cut} constraints 
and sophisticated presolving.

Although neither CP-SAT nor MIP solvers scale effectively for our use case, we can still employ them
for finding good upper and lower bounds 
by making them part of a \emph{Large Neighborhood Search} (LNS)~\cite{pisinger2019large},  
a local optimization heuristic in which we repeatedly aim to replace a part of the
current solution via an optimization problem.
This is also called \emph{destroy} and \emph{repair}, as we essentially destroy a part of a solution by erasing some variable assignments, and then repair it via an optimization that tries to find the best values for the freed variables.
For example, for our problem the destroy operation will delete a number of configurations from the solution, and the repair operation will try to find a minimal set of configurations to cover the now-missing interactions again.
We can scale the extent of this destruction and the size of the resulting
optimization problem for repair so that it can be solved to (near-)optimality using
CP-SAT or MIP.  The advantage of this approach over others
(e.g., Genetic Algorithms or Simulated Annealing) is that we do
not sequentially go through a pool of solutions, but can search in an
exponentially-sized space of improvements for the best solution, as they are
encoded implicitly by the repair problem.
There are generic implementations of LNS, such as
Relaxation Induced Neighborhood Search~\cite{danna2005exploring}, that have
become important parts of modern MIP solvers; using
domain-knowledge, one can engineer even more powerful implementations, 
as we demonstrate in this paper.

\subsection{Duality and Quality Certificates}
\label{subsec:duality}

As described in the introduction, an essential concept utilized in this paper 
is \emph{duality}. 

\begin{definition} 
\label{def:duality}
Let $\mathcal{A}: \min\{a(x)\mid x\in X\}$ be
a (primal) minimization problem with feasible set $X$ and cost function $a$.
Then a maximization problem $\mathcal{B}: \max\{b(x)\mid y\in Y\}$ 
with feasible set $Y$ and cost function $b$ is \emph{dual}
to $\mathcal{A}$, if for any $x\in X$ and any $y\in Y$, 
the objective values $a(x)$ of $x$ and
$b(y)$ of $y$ satisfy $a(x)\geq b(y)$.
\end{definition}

For a classical example of duality, consider the problem \textsc{Minimum Vertex
Cover}, a special case of \textsc{Set Cover}, which aims at finding a smallest
subset of vertices $S\subseteq V$ in a graph $G=(V,E)$, such that any edge $e\in
E$ contains one of the vertices in $S$, and the problem \textsc{Maximum
Matching}, which asks for a maximum set $M\subseteq E$ of disjoint edges. If $S$
is \emph{any} feasible vertex cover and $M$ is \emph{any} feasible matching,
then no two edges in $M$ can contain the same vertex in $S$, so 
$S$ must contain at least one separate vertex for each edge in $M$,
and $|S|\geq |M|$.) 

By definition, any feasible solution for the dual provides a lower bound
for the primal (i.e., the original) problem, and thus a quality
certificate. Moreover, a suitable solution to the dual can provide a short 
certificate of optimality for a suitable solution 
to the primal problem.

\begin{lemma} 
\label{lem:duality}
Consider a primal problem $\mathcal{A}: \min\{a(x)\mid x\in X\}$
and a dual problem $\mathcal{B}: \max\{b(x)\mid y\in Y\}$.
Let $x^*\in X$ and $y^*\in Y$ such that $a(x^*)=b(y^*)$.
Then $x^*$ is optimal for $\mathcal{A}$.
\end{lemma} 

\begin{proof}
Consider any feasible solution $x\in X$ of $\mathcal{A}$.
By assumption, $\mathcal{B}$ is dual to $\mathcal{A}$,
so $a(x)\geq b(y^*)=a(x^*)$, proving that $x^*$ optimal.
\end{proof}

If the objective value of a dual $\mathcal{B}$ coincides with the 
primal optimum (and thus, provides a certificate of optimality),
the lower bound provided by $\mathcal{B}$ is \emph{tight}; 
otherwise, we have a \emph{duality gap}. An example 
arises for \textsc{Minium Vertex Cover} and 
\textsc{Maximum Matching} if the graph $G$ is
bipartite (so it does not contain an odd cycle), where both optima 
always coincide; this is called \emph{strong duality}.
On the other hand, we may get a duality gap for non-bipartite graphs 
such as a triangle (for which a vertex cover requires two vertices, 
while a matching cannot contain more than one edge), so
we only have \emph{weak duality}.

For \NP-hard optimization problems, the existence of strong duality would imply
$\NP = \coNP$, which is considered similarly improbable as $\P = \NP$.  In this
paper, we explore the concept of mutual exclusiveness, defining a clique
problem (or an independent set problem on the dual graph) as a weakly dual
problem to our interaction sampling problem.  This duality provides mathematically
rigorous quality and optimality certificates for our empirical evaluation.  

\section{Lower bounds: Certifying the quality of solutions}
For a specific instance $M$ of a minimization problem,
a lower bound, $\lbValue \in \mathbb{R}$, ensures that no solution to $M$
can have an objective value less than $\lbValue$. For an \NP-hard
minimization problem (such as $t$-ISP), one cannot expect a general 
method for computing tight lower bounds, as this would amount
to membership in \coNP. However, this does not rule out
the possibility of determining good (and even tight)
lower bounds for large sets of relevant benchmark instances.

In this section, we show how we can obtain such useful lower bounds for the
$t$-ISP.
A key observation is the fact that there may be some pairs of interactions
$I,J\in \validInteractions$ that do not both occur in any valid configuration.  We call such interactions
\emph{mutually exclusive}; a whole set
$\exclusiveInteractions\subseteq \validInteractions$ of interactions is mutually exclusive 
if all pairs within are.

The notion of ``mutual exclusiveness'' has previously appeared in some related literature containing SAT-based approaches for combinatorial testing \cite{ansotegui2022incomplete, yamada2015optimization}.
Given a set of $t$ (not necessarily binary) feature variables, assigning different values to these variables results in distinct interactions.
No single configuration can cover any two of these interactions simultaneously.
Thereby, maximum number of such possible (valid) interactions over a fixed set of $t$ feature variables provides a simple lower bound.
While effective for non-binary features of large domains, often offering lower bounds close to optimal sizes \cite{ansotegui2022incomplete}, its application to binary features is limited, yielding a maximum lower bound of $2^t$.
Our work extends ``mutual exclusiveness'' to any pair of interactions that cannot coexist in a single configuration, potentially offering improved lower bounds for both binary and non-binary features.

Our lower bound is based on the observation that any interaction in a set of
mutually exclusive and valid interactions, $\exclusiveInteractions$, has to be part
of a different configuration in any complete $t$-wise interaction sample.
Hence, any sample needs at least $|\exclusiveInteractions|$ configurations.

\begin{theorem}
The problem of finding a maximum-cardinality set of mutually exclusive \(t\)-wise interactions is (weakly) dual to finding a minimum-cardinality complete \( t \)-wise interaction sample.
\end{theorem}
\begin{proof}
We prove this by contradiction.
Assume there exists a set \( E\subseteq \validInteractions \) of mutually exclusive interactions and a complete \( t \)-wise interaction sample \( S \) such that \( |E| > |S| \).
Since \( S \) must cover all feasible interactions \(\validInteractions\), including those in \( E \), there must be a feasible configuration \( C \in S \) that contains two interactions \( I, J \in E \) (by the pigeonhole principle).
Since \( I \) and \( J \) are mutually exclusive, \( C \) cannot be feasible, which is a contradiction and concludes the proof.
\end{proof}

For some pairs of interactions, mutual exclusiveness can be established
efficiently, e.g., if a clause from $\clauses$ is violated once all
literals from $I \cup J$ are part of the configuration, or that there exist two
contradicting literals $\literal \in I$ and $-\literal \in J$.  In theory,
deciding mutual exclusiveness is \coNP-hard and can require an expensive call to a SAT
solver to determine that there cannot exist a valid configuration $C\in ||\model{}||$ with $I \cup
J\subseteq C$.
In practice, we can mitigate this complexity, as false negatives are not critical for correctness; at worst, they may only weaken the proven bound by unnecessarily excluding interactions from the set $E$.

Identifying a maximum-cardinality set of mutually exclusive interactions is
closely related to several well-known \NP-hard problems, such as \textsc{Clique} and \textsc{Independent Set}~\cite{karp1972reducibility}. As it is 
\NP-hard itself, even if the complexity of detecting mutual
exclusiveness is ignored, it is non-trivial to obtain good lower bounds 
according to the idea outlined above.

In the following, we carefully characterize different sources of 
mutual exclusiveness among interactions. This, in turn, describes the space of all possible sets of mutual exclusive interactions.
Later we discuss how to efficiently and effectively explore such space of lower bound generating sets and  close the section with an example.

\subsection{Characterizing Mutual Exclusiveness}
\label{sec:mutual-exclusiveness}
In this section, we elaborate on the notion of mutual exclusiveness of
$t$-wise interactions and characterize different cases in which it occurs among
valid interactions. For brevity, the following notation assumes $t = 2$;
this does not impede the validity for other interactions (i.e., $t = 1$ or $t \geq 3$).  With the aid of our
characterization, the problem of finding a maximum cardinality set of mutually
exclusive interactions can be formulated as a binary program, a special case of
Mixed Integer Programming.  This allows us to approach the computation of good
lower bounds on the sample size using MIP solvers.

Consider a pair of valid interactions $I$ and $J\in \validInteractions$.  $I$ and $J$ may be mutually
exclusive, because there is an \emph{invalid} interaction $\{p,q\}$ with $p \in
I$ and $q \in J$, i.e., a combination that cannot appear together in any valid
configuration ($\forall C\in ||\model{}||: \{p,q\}\not\subseteq C$).
While this is not always the case, it explains the simplest form of mutual
exclusiveness that can be quickly detected.
Let us capture this form as the following property.
\begin{property} \label{pr:triangle-square}
A pair of valid interactions $I$ and $J\in \validInteractions$ satisfying
\[\exists p \in I ~\exists q \in J : \{p,q\} \notin \validInteractions\]
are mutually exclusive.
\end{property}

Mutual exclusiveness of the interactions $I$ and $J$ can occur even if the
above property does not hold.  Consider a pair $I,J$ of interactions that do
not satisfy \cref{pr:triangle-square} and a feature literal, $\literal \notin I
\cup J$.  If both $\{\literal\} \cup I \cup J$ and $\{-\literal\} \cup I \cup
J$ include some invalid interaction, $I$ and $J$ are also mutually exclusive since no valid configuration including $I\cup J$ could exist containing neither $\{\literal\}$ nor $\{-\literal\}$.

The same line of reasoning can be generalized to a generic exclusiveness
argument as follows.  For a subset of features $F\subseteq \mathcal{F}$, let $A(F)$ be the set of all
possible assignments of values to these features, with each assignment being 
a set of literals that contains a literal for each feature $x \in F$.
For example, $A(\{x\}) = \{\{x\},\{-x\}\}$ and $A(\{x,y\}) = \{\{x,y\}, \{-x,y\}, \{x,-y\}, \{-x,-y\}\}$.

Given a pair of valid interactions $I$ and $J$, let $\overline{\mathcal{F}}_{I,J}$
denote the set of all features {not yet assigned a value} by the assignment $I \cup J$.
Then $I$ and $J$ are mutually exclusive if there is a \emph{blocking} subset
of features $F \subseteq\overline{\mathcal{F}}_{I,J}$ for which any possible
extension of $I \cup J$ that also assigns values to the features $F$ contains an invalid interaction, as follows.
\begin{align*}
    P(I,J) &:=&  
    &~\exists F \subseteq \overline{\mathcal{F}}_{I,J} ~ \forall L \in A(F): & 
    \exists \{p, q\} \subset L \cup I \cup J \text{ with } \{p,q\} \notin \validInteractions. &
\end{align*}
%
%
%
Satisfying $P(I,J)$ means that the partial assignment $I \cup J$ cannot be
extended to a valid configuration.  However, there may still be some cases of
mutual exclusiveness that cannot be discovered by predicate $P(I,J)$.
To capture these additional cases, consider the situation in which predicate $P(I,J)$, as defined above, is false.
It means that there is an assignment $L \in A(\overline{\mathcal{F}}_{I,J})$, such that all $t$-wise interactions in $I \cup J \cup L$ are valid.
We can extend $P(I,J)$ to a predicate $D(I,J)$ that recognizes all cases of mutual exclusiveness by checking whether all such extensions $I \cup J \cup L$ are invalid configurations.
Note that the predicates $P(I,J)$ and $D(I,J)$ consider a large number of cases and can thus be costly to evaluate.
However, weaker forms of $P(I,J)$ can be evaluated by only searching for blocking sets $F$ of
limited size, e.g., of size $0$ (resulting in \cref{pr:triangle-square}),
$1$ or $2$, which can be done efficiently.
Now we can formulate the problem of finding a maximum cardinality set of mutually exclusive interactions as a binary program.
\begin{align}
\max \qquad & \sum_{\interaction \in \validInteractions} x_{\interaction} \label{eq:cds-ip:obj}\\
\forall I, J \in \validInteractions\text{ with }\neg Q(I, J): \qquad & x_{I} + x_{J} \leq 1, \label{eq:cds-ip:valid}\\
\forall I \in \validInteractions: \qquad & x_{I} \in \{0,1\}.\label{eq:cds-ip:vars}
\end{align}
Here $Q(I,J)$ is a predicate indicating mutual exclusiveness of interactions
$I, J$, e.g., $D(I,J)$ or, more practically, a restricted version of
$P(I,J)$ such as \cref{pr:triangle-square}.  In this formulation, the binary
variable $x_I$ decides whether the interaction $I$ is included in a desired set
of mutually exclusive interactions.  \cref{eq:cds-ip:obj} maximizes the number
of interactions in the desired set, while \cref{eq:cds-ip:valid} enforces
mutual exclusiveness for any pair of selected interactions. Here the set of
all valid interactions, $\validInteractions$, can efficiently be extracted from
a given feasible sample.  For any set of valid interactions
$\validInteractions'$, let $\lbIP(\validInteractions')$ denote the set of all
interactions selected by the binary program.

In practice, the above binary program can be very hard to solve, mainly due to
the large number of valid interactions $\validInteractions$ in many practically
relevant instances.  Therefore, we consider restricting the above
predicates and the binary program to a smaller subset of valid interactions
$\validInteractions' \subseteq \validInteractions$. 
By using the above binary program --- with \cref{eq:cds-ip:valid} generated by \cref{pr:triangle-square} ---
to optimize over different subsets of valid interactions,
we can construct a Large Neighborhood Search algorithm that is very effective.
We elaborate on the details of this algorithm later in this section.

\subsection{Extended Use of Invalid Interactions}\label{sec:start-solution}

The role of invalid interactions is not limited to detecting mutual
exclusiveness of valid interactions.  Invalid interactions also suggest a
constructive search pattern 
for heuristically finding large sets of mutually exclusive interactions as follows.

Consider a set of mutually exclusive interactions $\exclusiveInteractions$, and
a feature literal $\literal$. The set $\exclusiveInteractions$ is called
\emph{feature fixed} at $\literal$, denoted as
$\exclusiveInteractions_{\literal}$, if all interactions $I$ in
$\exclusiveInteractions$ contain $\literal$, i.e., $\forall I \in
\exclusiveInteractions: \literal \in I$.  Feature-fixed sets of mutually
exclusive interactions are of interest, because \cref{pr:triangle-square} immediately
provides the following.

\begin{property}\label{pr:invalidTransaction}
For any invalid interaction $\{p,q\}$, any feature-fixed sets
$\exclusiveInteractions_p$ and $\exclusiveInteractions_q$ are disjoint and
$\exclusiveInteractions_p \cup \exclusiveInteractions_q$ is mutually exclusive.
\end{property}

Starting at a feature literal $p$ and having a feature-fixed
$\exclusiveInteractions_p$ computed, this property suggests considering any
possible invalid interaction $\{p,q\}$ to extend $\exclusiveInteractions_p$ by
adding interactions from any $\exclusiveInteractions_q$. 

As a simple example, consider three features $\features=\{1,2,3\}$ and no constraints between them.
Thus, the only invalid interactions are of the form $\{i,-i\}$ for each $i \in \features$.
According to \cref{pr:invalidTransaction}, one could simply start with $E_{1} = \{ \{1,2\}, \{1,-2\} \}$, a mutually exclusive set of interactions fixed at feature $1$, and then be guided by an invalid interaction like $\{1,-1\}$ to safely report $E_{1} \cup E_{-1}$ as a bigger set of mutually exclusive interactions.
Here this new set could be $E = \{ \{1,2\}, \{1,-2\}, \{-1,2\}, \{-1,-2\} \}$; it is easy to verify that these interactions are mutually exclusive.
Furthermore, the sample $S = \{\{1,-2,3\},\{-1,-2,-3\},\{1,2,-3\},\{-1,2,3\}\}$ covers all valid pairwise interactions with four configurations.
Hence, for this example, it gives us a matching lower and upper bound on the
number of configurations in a sample, showing that both $E$ and
$S$ have optimum cardinality.

To obtain large mutually exclusive sets of interactions
by combining feature-fixed sets, it is sensible to focus on large
feature-fixed sets $\exclusiveInteractions_p$, $\exclusiveInteractions_q$.
Given a feature literal $\literal$, 
a maximum cardinality $\exclusiveInteractions_{\literal}$ could be obtained by a simple
adaptation of the above-defined binary program.
From another point of view, any maximum cardinality set of mutually exclusive
interactions can be written as the union of carefully selected feature-fixed
$E_{\literal}$'s over all $\literal$.

Now we describe an algorithm called $\lbIS$ to search for large
sets of mutually exclusive interactions.  We start with an empty set
$\exclusiveInteractions$ and repeat the following steps over each feature
literal $p$.  We first construct the maximum cardinality
$\exclusiveInteractions_p$ and try to add it to $\exclusiveInteractions$.  For
any invalid interaction $\{p,q\}$, we then construct the maximum cardinality
$\exclusiveInteractions_q$ and try to add $\exclusiveInteractions_q$ to
$\exclusiveInteractions$.  During the above two steps, it may be impossible to
add some interaction $I$ to $\exclusiveInteractions$ because $I$ may not be
mutually exclusive with some other interaction $J$ already in
$\exclusiveInteractions$. In this case, the conflicting interaction $J \in
\exclusiveInteractions$ is recorded.
After a pre-specified number of adding operations, with some fixed probability, either the highest conflicting portion of $\exclusiveInteractions$ is removed, or $\exclusiveInteractions$ is emptied.

\subsection{Large Neighborhood Search for Lower Bounds}

While $\lbIS$ often yields large sets of mutually exclusive interactions, it can still be improved.
In the following, we describe how $\lbIP(\validInteractions')$, which solves the binary program defined by \cref{eq:cds-ip:obj,eq:cds-ip:valid,eq:cds-ip:vars}, can be utilized in a large neighborhood search to obtain even larger sets.

\alg{LB-LNS} (\cref{alg:lb-lns}) starts with an initial solution from the lower bound heuristic in
\cref{sec:start-solution}.  Subsequently, it iteratively removes some part
$\exclusiveInteractions'$ of the current best solution and replaces it with an
optimal solution from $\lbIP$.  Computing the candidate interactions
$\validInteractions'$, i.e., the set of all interactions that are mutually
exclusive to the remaining elements $\exclusiveInteractions \setminus
\exclusiveInteractions'$, can be done efficiently during the selection of
$\exclusiveInteractions'$.  At the end of each iteration, a new solution is
constructed by taking the union of $\exclusiveInteractions \setminus
\exclusiveInteractions'$ with the output of $\lbIP(\validInteractions')$.
The formal description is given in \cref{alg:lb-lns} and an example can be found in \cref{sec:lb:example}.

\begin{algorithm}[t]
\begin{algorithmic}
\STATE $\exclusiveInteractions=\lbIS()$   \hfill \textit{[Initial set of mut. excl.\ interactions]}
\WHILE{within time limit}
\STATE Select $\exclusiveInteractions'\subseteq \exclusiveInteractions$ \hfill \textit{[The part of the set we remove]}
\STATE $\validInteractions'=$ valid interactions that are mut.\ excl.\ to  $\exclusiveInteractions \setminus \exclusiveInteractions'$
\STATE $\exclusiveInteractions''=\lbIP\left(\validInteractions'\right)$\hfill \textit{[With short time limit]}
\STATE $\exclusiveInteractions=\left(\exclusiveInteractions\setminus \exclusiveInteractions'\right) \cup \exclusiveInteractions''$ \hfill \textit{[Combine to new set]}
\STATE Tune $\exclusiveInteractions'$-selection based on difficulty of computing $\exclusiveInteractions''$
\ENDWHILE
\RETURN $\exclusiveInteractions$
\end{algorithmic}
\caption{LB-LNS}\label{alg:lb-lns}
\end{algorithm}

The most challenging part of the algorithm is the selection of interactions to remove from a current solution.
The size of the remaining interactions $\validInteractions'$ must be small
enough so that $\lbIP$ can still solve it, but large enough to yield
any improvement for the current best solution.  Thus, the selection
procedure for $\exclusiveInteractions'$ is adapted according to the difficulty
of the integer program subroutine.
That is choosing a smaller or larger $E'$ for the next iteration.

\subsection{Example}\label{sec:lb:example}

For a small example of \cref{alg:lb-lns}, consider a mutually exclusive set of interactions
with the features $\mathcal{F}=\{1,2,3\}$ and a single clause
$\mathcal{D}=\{\{-1,-3\}\}$.  
The \emph{Compatibility Graph}, as shown in \cref{fig:lb:example}(a), has a vertex for each valid interaction and edges for pairs of interactions
that can appear in the same configuration. As a consequence, 
we are looking for an independent set, i.e., a subset
of vertices without any edges in between.
We start with a randomly obtained initial set of mutually exclusive interactions
$E=\{\{1,2\}, \{1, \textrm{-}2\}, \{\textrm{-}1, 3\}, \{\textrm{-}1,
\textrm{-}3\}\}$ as highlighted in red in \cref{fig:lb:example}(b).
We can see that it is indeed mutually exclusive as it is an independent set without any edges between its elements.
Now perform an iteration of LNS, by (randomly) selecting $E'=\{\{\textrm{-}1,
3\}, \{\textrm{-}1, \textrm{-}3\}\}$ as the subset to be removed, see
\cref{fig:lb:example}(c).  
This gives us the 
smaller subproblem with 
$\validInteractions'=\{\{\textrm{-}1, 2\},
\{\textrm{-}1, \textrm{-}2\}, \{\textrm{-}1, 3\},$ 
$ \{\textrm{-}1,\textrm{-}3\},
\{2, 3\}, \{\textrm{-}2, 3\}\}$, see \cref{fig:lb:example}(d).  This subproblem
consists of all interactions that are mutually exclusive with the whole
remaining set; usually this strongly reduces the number of vertices.
We compute on $\validInteractions'$ an optimal mutually exclusive set
$E''=\{\{\textrm{-}1, \textrm{-}3\},$
$ \{2,3\}, \{\textrm{-}2, 3\}\}$,
see \cref{fig:lb:example}(e).  Because $\validInteractions'$ is mutually exclusive
to the remaining set $E\setminus E'$, we can combine $\left(E\setminus
E'\right)\cup E''$ to a new, larger set, see \cref{fig:lb:example}(f).  This
could be repeated; in the example, we have already reached optimality.

\begin{figure*}
\centering
\includegraphics[page=2,width=.98\textwidth]{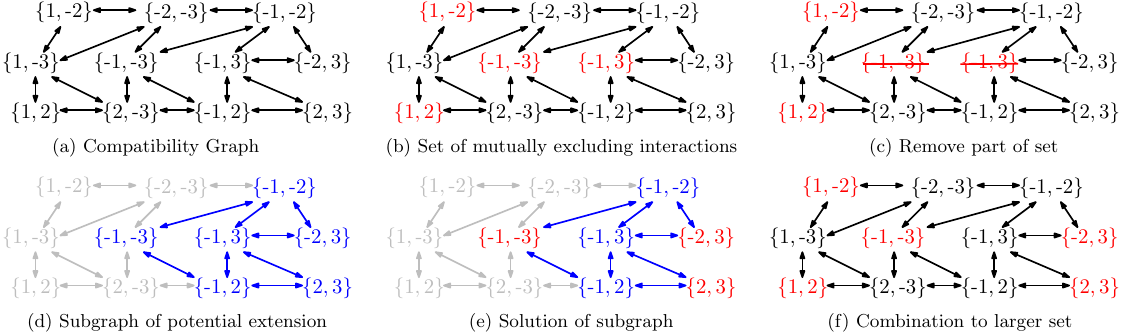}
\caption{
Example of the lower bound computation with \cref{alg:lb-lns}.
An edge in the compatibility graph indicates that the two interactions are compatible and can appear in the same configuration.
A set of interactions without any edges, like in (b), is mutually exclusive and a lower bound on the necessary number of configurations.
In this example, the initial set of mutually exclusive tuples $\{\{1,2\}, \{1, \textrm{-}2\}, \{\textrm{-}1, 3\}, \{\textrm{-}1,
\textrm{-}3\}\}$ increases to $\{\{1,2\}, \{1, \textrm{-}2\}, \{\textrm{-}1, \textrm{-}3\},$
$ \{2,3\}, \{\textrm{-}2, 3\}\}$.
}
\label{fig:lb:example}
\end{figure*}
\section{Computing (near-) optimal solutions}\label{sec:upper-bounds}

In the previous section, we explored the computation of large sets of mutually exclusive interactions,
which provide lower bounds on the achievable sample size but do not yield the samples themselves.
In this section, we focus on minimizing the actual sample size.

Both the lower bound computation and the sample minimization will be executed in parallel.
These processes will continue until they converge to the same value,
indicating an optimal sample size, or until a specified time limit is reached.
Due to the inherent duality of these problems, we can still bound the optimality gap,
even if convergence is not achieved within the time limit.

We begin by discussing the computation of optimal solutions for small instances,
which we achieve by modeling the problem as a (CP-)SAT formula.
This formulation is then used as a subroutine to compute (near-)optimal solutions for larger instances through a Large Neighborhood Search (LNS) approach.
In LNS, we iteratively solve sufficiently small parts of a larger instance to optimality,
progressively improving the overall solution quality.

\subsection{Modeling Pairwise Sampling in CP-SAT}\label{sec:upper-bounds:cpsat}

Finding a minimal sample to cover all valid interactions $\validInteractions$ can be described by a set of Boolean 
and linear constraints that can be solved by the constraint programming solver CP-SAT.
This does not scale well, but is a valuable subroutine in the later algorithm.

Let $x_1, x_2, \ldots, x_n\in \mathbb{B}$ represent the feature selection. 
Recall that a clause $\clause_j$ is a set of (negated or unnegated) literals, indicating 
the desired feature values.
For ease of notation we denote for any literal $\literal \in \clause_j$, $x_\literal$ to be the (possibly negated) 
representation of the feature $|\literal|$, i.e.
 $x_\literal = \overline{x_{|\literal|}}$ if $\literal < 0$.
We introduce further variables $y_\interaction \in \mathbb{B}$ for all interactions $\interaction \in \validInteractions$, 
which can only be true if the interaction $\interaction$ is covered by the assignments of $x_1, x_2,\ldots x_n$.
For an interaction  $\interaction$, this can be enforced by adding implications $y_\interaction \implies \bigwedge_{\literal \in \interaction} x_\literal$.

The idea of our CP-SAT formulation is to use $k$ copies of variables and constraints, where $k$ 
is an upper bound on the expected number of configurations.  If an
estimate for $k$ is too low, the model becomes
infeasible, and the solution process can be retried for increased $k$.

We can find a feasible sample with at most $k$ configurations by enforcing that for each interaction $\interaction$, at least one of the $k$ copies of $y_\interaction$ must be true.
For each of the $k$ copies, we introduce a further Boolean variable $u_i\in \mathbb{B}$ that indicates if the
copy is used to represent a sample configuration.
To minimize the number of sample configurations required to cover all interactions, we can then use the objective
to minimize the number of copies used.

\begin{align}
\min \quad & \sum_{i=1}^k u_i \label{eq:cp:obj}\\
\forall{i=1,\ldots k, \interaction \in \validInteractions}:\quad  & \overline{u_i} \implies \overline{y^i_\interaction} \label{eq:cp:y}\\
\forall{i=1,\ldots k, \clause_j \in \clauses}:\quad & \bigvee_{\literal \in \clause_j} x^i_\literal \label{eq:cp:conf}\\
\forall{i=1,\ldots k, \interaction \in \validInteractions}:\quad &y^i_\interaction \implies \bigwedge_{\literal \in \interaction} x^i_\literal  \label{eq:cp:cov}\\
\forall{\interaction \in \validInteractions}: \quad & \bigvee_{i=1}^k y^i_\interaction \label{eq:cp:t}
\end{align}

Here \cref{eq:cp:obj} minimizes the sample size, i.e., the number of active copies (indicated by $u_i\in \mathbb{B}$), 
\cref{eq:cp:y} prevents deactivated copies of sample configurations from covering interactions,
\cref{eq:cp:conf} ensures sample configurations are valid according to the feature model, 
\cref{eq:cp:cov} forces $y^i_\interaction$ to be false if the interaction $\interaction$ is not covered in the $i$-th configuration, 
and \cref{eq:cp:t} ensures each interaction $\interaction \in \validInteractions$ to be covered by at least one configuration.

For simplicity, we restrict ourselves to CNF; the solver CP-SAT actually allows complex constraints with more efficient formulations.
We use these more efficient constraints whenever possible in the actual implementation, especially for modeling the common ``alternative''-constraint.
However, if the instance can actually be expressed efficiently in CNF, using a SAT solver with iterative search for optimizing the objective, as done by Yamada et al.~\cite{yamada2015optimization}, may be more efficient than CP-SAT if the model only utilizes a CNF representation.
While the formulation can theoretically also be solved with Integer Programming, the strong Boolean logic of the problem makes it more suitable for SAT-based solvers as they are specifically designed for this.

A serious issue with this formulation is its high symmetry due to the copied parts:
Every solution, feasible or infeasible, has a number of representations exponential in $k$ 
by permuting the assignments of $x^i_1, \ldots, x^i_n$ over $i$.
This difficulty for the solver
is common, e.g., also for {\sc Graph Coloring}, for which permuting the
colors yields an exponential number of different but logically equivalent
solutions (cf.~\cite{law2006symmetry}).
We can break these symmetries and
greatly improve the performance of the solver 
by computing a set of mutually exclusive interactions $\exclusiveInteractions$ using \cref{alg:lb-lns}
and enforce each interaction of that set to be covered in a specific copy $i$.
\begin{equation}
\forall \interaction \in  \exclusiveInteractions: \quad y_{\interaction}^i = \texttt{true} \label{eq:symmetry}
\end{equation}

In the following, the computed optimal sample is denoted by $\text{\textsc{OptSample}}\left(\features, \clauses, \validInteractions, k\right)$.
If there is a time limit, it will return the best solution found, or $\bot$ if no sample 
that covers all interactions in $\validInteractions$ has been found.

When an initial solution is available -- whether precomputed or quickly obtained through one of the many available heuristics -- it can significantly assist \textsc{OptSample}.
This initial solution serves multiple purposes: it provides a sufficient value for $k$, offers a warm start for the solver,
and acts as a fallback, ensuring that a feasible solution is always returned,
even if \textsc{OptSample} does not terminate within the allotted time.

With this approach, most of our benchmark models with fewer than \num{20} features could be solved directly to provable optimality.
For larger models, the formulation size of $\Omega\left(n\times k +
|\validInteractions |\times k\right)$ becomes quickly prohibitive (mostly due
to the quickly growing $|\validInteractions |$).

\subsection{Large Neighborhood Search}

The preceding ideas can be extended to compute (near-)optimal solutions for even larger configurable systems.
Specifically, if we already have a sample,
we can use \textsc{OptSample} to optimize a large portion of it by a significantly smaller, practically tractable optimization model.
This approach allows us to efficiently explore parts of interest in the solution space with each iteration,
leading to rapid convergence toward a (near-)optimal solution.

Although \cref{eq:cp:conf} is the most complex constraint, it can generally be solved quickly, as demonstrated by algorithms like YASA~\cite{KTS+:VaMoS20}, ICPL~\cite{MHF:SPLC12}, IncLing~\cite{AKT+:GPCE16}, and Chvátal's algorithm~\cite{JHF:MODELS11}.
However, the crucial challenge lies in managing the number of additional constraints that arise from the large set of valid interactions $\validInteractions$.
To keep the optimization model in \textsc{OptSample} practically tractable, we need to reduce the size of $\validInteractions$.

An important observation is that the coverage of $\validInteractions$ increases rapidly with each additional configuration,
while only a small amount of coverage is lost even when multiple configurations are removed from a sample.
This behavior is illustrated for some concrete configurable systems in \cref{fig:coverage_examples}.

If we remove a subset $S' \subset S$ from a sample $S$, uncovering the interactions $\validInteractions' \subseteq \validInteractions$, we can use $\text{\textsc{OptSample}}\left(\features, \clauses, \validInteractions', |S'|\right)$ to repair the sample with as few configurations as possible.
For instance, in the \texttt{FreeBSD-8\_0\_0} sample shown in \cref{fig:coverage_examples},
removing \SI{50}{\percent} of the sample results in a loss of only about \SI{0.3}{\percent} of the coverage.
Consequently, we can apply $\text{\textsc{OptSample}}\left(\features, \clauses, \validInteractions', |S'|\right)$ to find the best replacement for \SI{50}{\percent} of the sample, using an optimization model that is only roughly \SI{0.3}{\percent} the size of the full problem.
This approach still provides significant optimization potential while drastically reducing the problem size.

We incorporate this idea into an algorithm (\cref{alg:samplns}) that iteratively removes parts of an initial solution and repairs it optimally.
An initial sample can be quickly computed using YASA~\cite{KTS+:VaMoS20} with the parameter $m=1$.
\begin{algorithm}[ht!] 
\begin{algorithmic}
\STATE $\sample=$ initial sample computed by YASA ($m=1$)
\STATE Start \textsc{LB-LNS} on a separate thread
\WHILE{within time limit \textbf{and} $|S|>$ lower bound}
\STATE Select $\sample'\subseteq \sample$ \hfill \textit{[The part of the solution we remove]}
\STATE $\validInteractions'=$ interactions covered by $\sample'$ but not by $\sample\setminus \sample'$
\STATE $k=|\sample'|$ \hfill \textit{[Upper bound on covering $\validInteractions'$]}
\STATE $\sample''=\text{\textsc{OptSample}}\left(\features, \clauses, \validInteractions', k\right)$ \hfill \textit{[short time limit]}
\STATE $\sample=\left(\sample\setminus \sample'\right) \cup \sample''$ \hfill \textit{[Combine to new solution]}
\STATE Tune $\sample'$-selection based on difficulty of computing $\sample''$
\ENDWHILE
\STATE Terminate \textsc{LB-LNS} and retrieve latest $E$
\RETURN optimized sample $\sample$, lower bound $|E|$
\end{algorithmic}
\caption{SampLNS}
\label{alg:samplns}
\end{algorithm}

A performance-critical part of \cref{alg:samplns} is the size of $\sample'$, the configurations that are removed 
from the current sample $\sample$ (similar to the selection of $E'$ in \cref{alg:lb-lns}).
This needs to be sufficiently small, such that 
$\text{\textsc{OptSample}}\left(\features, \clauses, \validInteractions', k\right)$ 
can compute an optimal (or at least a good) solution.
At the same time, it should be as large as possible to actually achieve improvements.
The difficulty of the resulting subproblem is hard to predict, so the selection should be dynamic and increase the size if the 
iterations are very fast, or decrease its size if the iterations are reaching the time limit.

As previously mentioned, \textsc{OptSample} can be supported by providing an initial solution,
ensuring that it will either return this solution or find a better one.
We use the removed part, $\sample'$, as the initial solution,
which guarantees that the new solution will either be an improvement upon or revert to $\sample$, since $\left(\sample \setminus \sample'\right) \cup \sample' = \sample$.

\subsection{Example}

For the benefit of the reader, we provide an iteration of the algorithm for a simple example.
Consider an instance with features $\mathcal{F}=\{1,2,3,4\}$ and two clauses $\mathcal{D}=\{\{1,2\}, \{3,4\}\}$.
This instance has \num{22} valid pairwise interactions $\validInteractions$ that need to be covered.
\begin{align*}
     \validInteractions = \{~& \{3, 4\}, \{1, \textrm{-}3\}, \{2, \textrm{-}4\}, \{1, 3\}, \{\textrm{-}2, 4\}, \{\textrm{-}1, 4\},  \{2, 4\}, \{1, 2\}, \{1, \textrm{-}4\}, \{\textrm{-}2, \textrm{-}3\}, \{\textrm{-}1, \textrm{-}3\}, \{\textrm{-}2, 3\}, \\
    &\{\textrm{-}1, 3\}, \{3, \textrm{-}4\}, \{\textrm{-}3, 4\}, \{2, \textrm{-}3\}, \{1, \textrm{-}2\}, \{1, 4\}~,  \{2, 3\}, \{\textrm{-}1, \textrm{-}4\}, \{\textrm{-}2, \textrm{-}4\}, \{\textrm{-}1, 2\} ~\}
\end{align*}
Assume that YASA has provided the following initial sample of size six:
\begin{align*}
\sample=\{&\{1,2,\textrm{-}3,4\}, \{1,\textrm{-}2,3,\textrm{-}4\}, \{1,\textrm{-}2,\textrm{-}3,4\}, \{\textrm{-}1,2,3,4\}, \{\textrm{-}1,2,3,\textrm{-}4\}, \{\textrm{-}1,2,\textrm{-}3,4\}\}.
\end{align*}
By looking at the interactions $\{1,2\}, \{1, 3\}, \{\textrm{-}2, \textrm{-}3\}, \{3,4\},$ $\{\textrm{-}1, \textrm{-}4\},$ and $ \{\textrm{-}1, \textrm{-}3\}$, it is straightforward to check that no configuration can be removed without missing an interaction, so there is no trivial improvement.
Now consider the three configurations 
\[\sample'=\{\{1,2,\textrm{-}3,4\},  \{\textrm{-}1,2,3,4\}, \{\textrm{-}1,2,3,\textrm{-}4\}\}\]
for possible removal.  This would leave the interactions
\[\validInteractions'=\{\{1,2\}, \{2,3\}, \{3,4\}, \{2,\textrm{-}4\}, \{\textrm{-}1,3\}, \{\textrm{-}1,\textrm{-}4\}\}\]
uncovered, for which we seek a cover with fewer configurations.
Solving 
\textsc{OptSample}$\left(\features, \clauses, \validInteractions', |\sample'|\right)$
gives us an optimal sample
$\sample''=\{\{1,2,3,4\}, \{\textrm{-}1,2,\textrm{-}3,4\}\}$ for $\validInteractions'$, which is smaller than the part we removed.
Notice how we removed half of the configurations, but only had to find a new sample for \num{6} out of the \num{22} interactions.
\Cref{fig:coverage_examples} shows that this effect is even stronger for real instances, explaining the effectiveness of our approach.
Combining yields a better sample for all interactions
\begin{align*}
\sample&=\left(\sample\setminus \sample'\right) \cup \sample''=\{~ \{1,\textrm{-}2,3,\textrm{-}4\}, \{1,\textrm{-}2,\textrm{-}3,4\},  \{\textrm{-}1,2,\textrm{-}3,4\},\{1,2,3,4\}, \{\textrm{-}1,2,\textrm{-}3,4\} ~\}.
\end{align*}
This is optimal for the example, but could be continued if necessary;
in that case, keeping track of the remaining computing time allows 
adjustment by selecting a smaller or larger $\sample'$ for the next iteration.

\begin{figure}
    \centering
    \includegraphics[width=0.4\columnwidth]{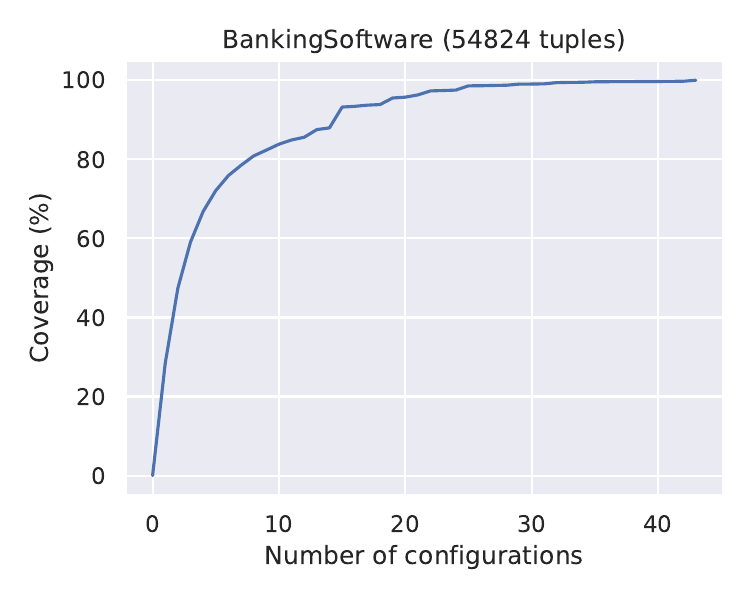}
    \includegraphics[width=0.4\columnwidth]{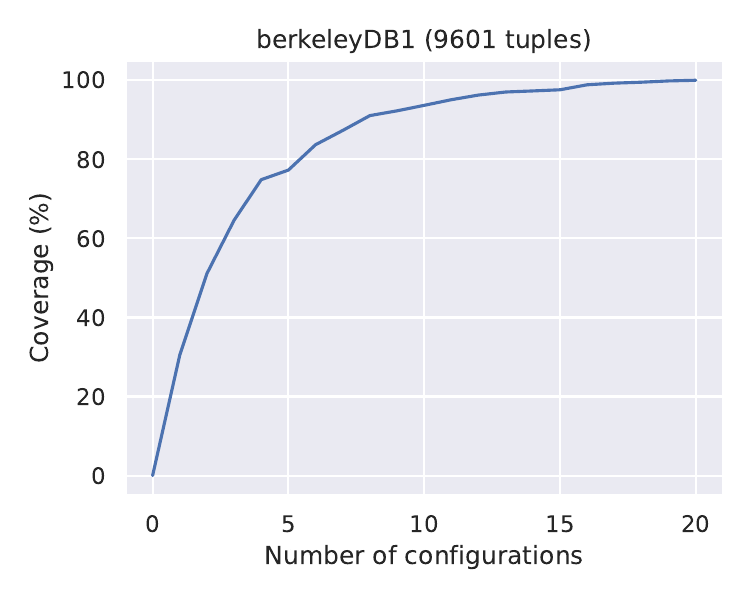}
    \includegraphics[width=0.4\columnwidth]{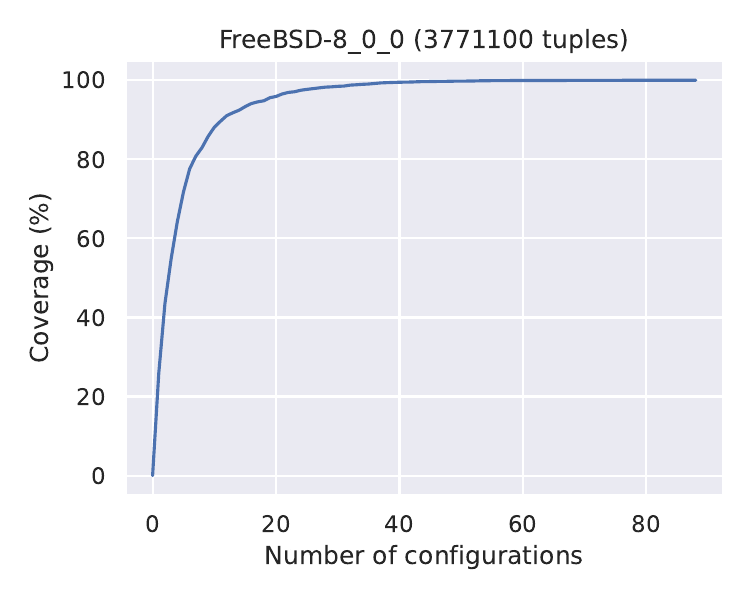}
    \includegraphics[width=0.4\columnwidth]{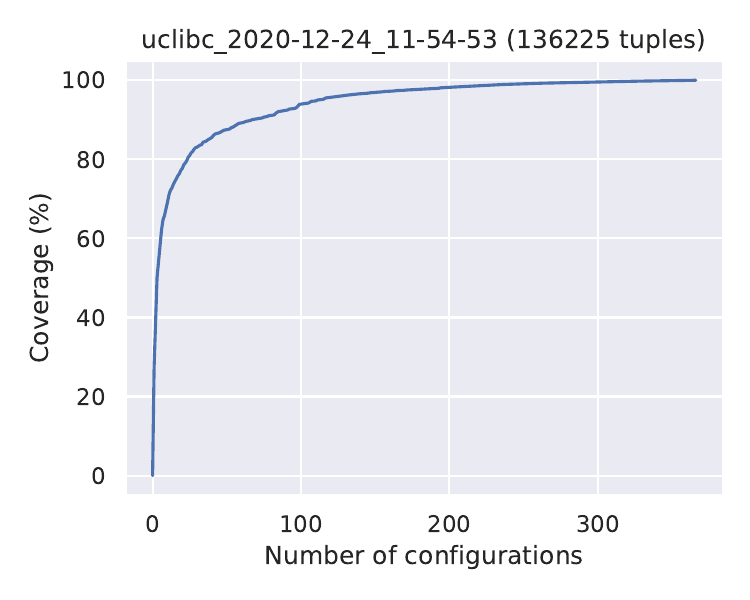}
    \caption{
    Examples of the coverage over the number of configurations.
    Each plot illustrates how coverage increases with the number of configurations:
    for instance, a value of \SI{90}{\percent} at \num{10} indicates that the first ten configurations in a randomly ordered sample would cover \SI{90}{\percent} of all feasible interactions.
    These samples were computed using YASA~\cite{KTS+:VaMoS20}.
    The rapid growth in coverage shows that even removing a large fraction of configurations leaves only a minor portion of interactions uncovered,
    which is crucial for the size of our constraint program.
    With few interactions becoming uncovered,
    the task of optimally recovering these through the constraint program becomes manageable.
    By shuffling the sample,
    we can repeatedly explore replacing large subsets of the sample with minimal configurations to achieve equivalent coverage.
    }
    \label{fig:coverage_examples}
\end{figure}

\section{Further Improvements}\label{sec:implementation}

The practical implementation can be tuned in various ways.

\paragraph{Instance Encoding and Preprocessing}

We accept feature models in the XML-format of FeatureIDE~\cite{MTS+17}, consisting of a feature hierarchy and a set of cross-tree constraints in propositional logic. 
This instance is simplified by some simple rules, e.g., replacing trivial clauses with only one element, or substituting literals with a bidirectional implication by the same variable. 
The feature hierarchy also defines the set of concrete features, whose interactions have to be covered.
Additionally, we also accept the DIMACS format of SAT solvers, which is a CNF formula with a set of variables and clauses.
This format is not preprocessed and all variables are considered to be concrete features.

\paragraph{Selecting \texorpdfstring{$\exclusiveInteractions'\subseteq \exclusiveInteractions$ in \cref{alg:lb-lns}}{E' from E for the lower bound}}

We need to select $\exclusiveInteractions'\subseteq \exclusiveInteractions$ in
\cref{alg:lb-lns}, such that we can still solve $\lbIP$, but as large as
possible to escape deeper local optima.  For this, we start with $E'=E$ and
iteratively remove random interactions from it, until $|\validInteractions'|$ is
below a threshold $\gamma=1000$.
Interactions that would result in $\validInteractions'=\emptyset$ are excluded from removal.
To scale the threshold $\gamma$ automatically to the right size, it is increased by \SI{10}{\percent} if an iteration uses less than \SI{50}{\percent} of its time,
 and decreased by \SI{10}{\percent} if it uses more than \SI{95}{\percent} of its time limit.


\paragraph{Selecting \texorpdfstring{$\sample'\subseteq \sample$}{S' from S} in SampLNS}

We use an analogous strategy for \cref{alg:samplns}.
Starting with $\sample'=\emptyset$, we repeatedly select a random configuration from $\sample\setminus \sample'$ and add it to $\sample'$,
as long
as the number of missing interactions in $\sample\setminus \sample'$ remains
below a threshold $\phi$.  Initially $\phi$ is set to \num{250}, but we
increase this value by \SI{25}{\percent} if \textsc{OptSample} computes an
optimal solution and decrease it again by \SI{25}{\percent} if
\textsc{OptSample} is not able to compute an optimal solution or improve the
sample.  Due to its exponential character, this simple trick allows us to 
automatically scale to the largest feasible size of $\sample'$ within a few iterations.

\paragraph{Symmetry Breaking in \text{\textsc{OptSample}}}

At the beginning of each iteration in \cref{alg:samplns}, we compute a set of
mutually exclusive interactions on $\validInteractions'$ for breaking
symmetries in \textsc{OptSample}, or skipping the optimization if we notice
that no improvement is possible.  For this, we first compute a simple greedy
solution (sorting based on the number of appearances in a sample, as
interactions with fewer appearances are more likely to be exclusive) and then
optimize it with \cref{alg:lb-lns} for up to \SI{10}{\percent} of the iteration time limit.

\paragraph{Parallelization}

Our implementation of SampLNS will spawn a thread running \cref{alg:lb-lns} (single threaded) and use all remaining cores to optimize the sample (CP-SAT can make use of multiple cores).
After each iteration of SampLNS, the latest lower bound of \cref{alg:lb-lns} is queried (used to check for premature termination due to optimality).
For smaller feature models, \textsc{OptSample} may yield better lower bounds than \cref{alg:lb-lns}, such that the lower bounds are in rare cases not actually computed by \cref{alg:lb-lns}.
The thread of \cref{alg:lb-lns} is independent of the additional use of \cref{alg:lb-lns} for symmetry breaking at the beginning of each iteration.

\section{Empirical Evaluation}
In the preceding sections, we presented a new parametric algorithm called SampLNS that is able to minimize samples and compute provable lower bounds for the achievable sample size.
In this section, we investigate the quality of the lower bounds and the minimization potential empirically.
We designed and performed an experiment to answer the following four research questions.

\newcommand{\rqone}{Can existing algorithms compute samples with a sample size equal or close to the lower bound}
\newcommand{\rqtwo}{Can SampLNS compute samples of smaller size than existing algorithms}
\newcommand{\rqthree}{What is the computational effort to derive lower bounds and optimize samples with SampLNS}
\newcommand{\rqfour}{Does SampLNS depend on the sample quality of YASA}

\begin{itemize}
    \item[RQ1] \rqone? 
    \item[RQ2] \rqtwo?
    \item[RQ3] \rqthree?
    \item[RQ4] \rqfour?
\end{itemize}

\subsection{Experiment Design}\label{subsec:experiment_design}

We used feature models from multiple sources that originate from different domains, namely finance, systems software, e-commerce, gaming, and communication~\cite{PTR+:SPLC19,PKR+:VaMoS21,KTS+:ESECFSE17,MTS+17,SLBWC:ICSE11} and encompass a wide range of number of features (9--1,408) and number of constraints (13--15,692) \cf{\autoref{tab:eval:new_benchmark_results}}.
For a wide range of feature model sizes, we selected small- and medium-sized feature models from examples provided by the tool FeatureIDE~\cite{MTS+17}.
We also used feature models from real-world Kconfig systems, provided by Pett~et~al.~\cite{PKR+:VaMoS21}, for which we chose the earliest and latest versions for each system.
In addition, we used more complex, real-world feature models~\cite{KTS+:ESECFSE17,SLBWC:ICSE11,PTR+:SPLC19}.
Knüppel~et~al.~\cite{KTS+:ESECFSE17} provide \num{116} different models for the eCos system.
Most of these models have similar properties, so we only used six models with different numbers of features.

To investigate the performance of existing algorithms and SampLNS, we selected multiple state-of-the-art algorithms for $t$-wise interaction sampling and compared their performance on the benchmark models.
In particular, we selected the following algorithms: \emph{Chv\'atal}~\cite{C:MOR79,JHF:MODELS11}, \emph{ICPL}~\cite{MHF:SPLC12}, 
IPOG~\cite{LK+:ECBS07}, 
\emph{IncLing}~\cite{AKT+:GPCE16}, and YASA~\cite{KTS+:VaMoS20}.
All of these are well-known greedy algorithms with different sampling strategies and were also used in previous evaluations~\cite{AKT+:GPCE16,ATL+:SoSyM19,MHF:SPLC12,KTS+:VaMoS20}.
Especially ICPL is regarded as one of the best $t$-wise sampling algorithms and was used as benchmark for many other algorithms that we did not include in this evaluation~\cite{GCD:ESE11,LK+:ECBS07,OZML:VaMoS11,HPP+:TSE14}.
Some other algorithms from previous evaluations either lack publicly available implementations or are incompatible with our Linux workstations~\cite{HLHE:VaMoS13,PSK+:ICST10,SCD:FASE12,CDS:TSE08,GCD:ESE11,OMR:SPLC10}, and therefore could not be included in our comparison.
However, based on prior evaluations, these missing algorithm are not expected to be competitive.

As our problem specification requires \SI{100}{\percent} pairwise coverage, we have chosen only algorithms that are guaranteed to generate such samples.
This excludes, for instance, algorithms that generate random samples~\cite{OGB:SPLC19} and most algorithms that employ local or population-based search~\cite{HPP+:TSE14,CR:IJAST14,EBG:CAiSE12,HPP+:SPLC13}.

YASA provides a parameter $m$ for decreasing the sample size through multiple sampling iterations.
We generate samples with YASA using multiple values for this parameter,
in particular, $1$, $3$, $5$, and $10$.
To ensure a more equitable comparison with SampLNS, which utilizes the full time limit for refinement,
we also implemented a variant of YASA that, instead of being constrained by a fixed number of sampling iterations \(m\),
performs as many iterations as possible within the given time limit.

We conducted three experiments as follows:
\begin{enumerate}
\item First, we establish provable lower bounds for each model by executing SampLNS five times with a generous time limit of \SI{3}{\hour} (\SI{180}{\second} for each optimization step) and using the best bounds obtained.
These lower bounds were subsequently employed in the other experiments to estimate the optimality gap.
\item Each of the previously mentioned algorithms, as well as SampLNS, was run five times for each model with a time limit of \SI{900}{\second} (\SI{60}{\second} for each optimization step).
The mean of the results was used to compare the performance of these algorithms.
For SampLNS, this time limit includes the runtime of YASA for obtaining the initial sample to ensure a fair comparison, even though YASA ($m=1$) typically completes in just a few seconds.
\item SampLNS was additionally run with four different algorithms for computing the initial sample, five times for each model, with a time limit of \SI{900}{\second} (\SI{60}{\second} for each optimization step).
This experiment aimed to assess the impact of the initial sample on the performance of SampLNS.
Importantly, the runtime of the initial sample computation was this time not counted within the time limit, ensuring that SampLNS had an equal opportunity for each initial sample.
\end{enumerate}

All experiments were executed on AMD~Ryzen~9~7900 
with \SI{96}{\giga\byte} DDR5 RAM and Ubuntu~22.04.
The code was written in Python~3.12 with native C++-17 elements using \mbox{PyBind11}~(v2.10.3), compiled by g++ (v11.3.0).
We used CP-SAT of or-tools~(v9.8) and the MIP solver Gurobi~(v11.0).
Code and data are available online (\url{https://github.com/tubs-alg/SampLNS}).
For the baseline algorithms and YASA, we use the FeatureIDE library, written in Java, which provides a uniform API to several sampling algorithms~\cite{AMK+:GPCE16,KPK+:SPLC17}.

\subsection{Results}

\subsubsection{(RQ1) \rqone} 
\begin{figure}
\centering
\includegraphics[width=0.6\columnwidth]{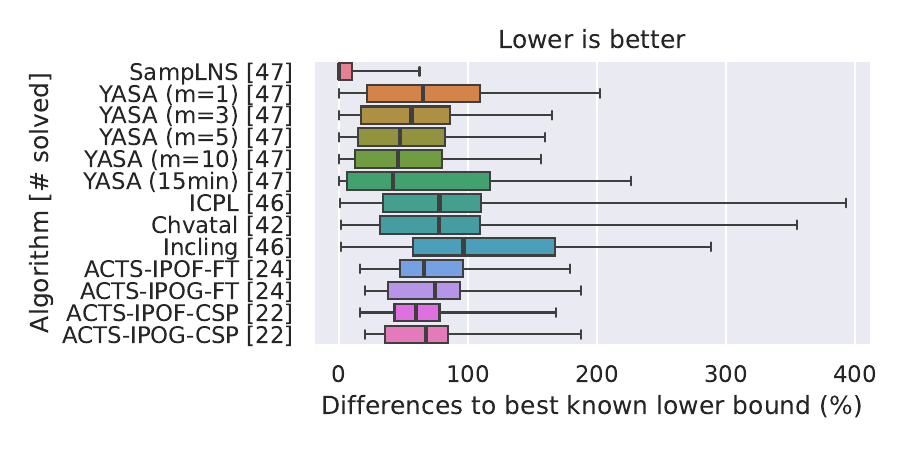}
\caption{
    Differences of the sample sizes of various sampling algorithms (with a \SI{900}{\second} time limit) to the best lower bound computed by SampLNS.
    Not all algorithms were able to compute a feasible sample within the time limit, thus, the number of successfully solved models is added in parentheses.
    We can see that SampLNS never timed out, and \SI{75}{\percent} of the samples are at most \SI{10}{\percent} above the lower bound.
    More than \SI{55}{\percent} of the samples match the lower bound and, thus, are minimal.
    The next best algorithm is YASA with $m=10$ (resp. the \SI{15}{\min}-variant) with a median difference of \SI{46}{\percent} (resp. \SI{42}{\percent}).
    Incling yields the largest samples, with a median of twice the size of the lower bound.
}\label{fig:eval:diff_to_lb}
\end{figure}
\begin{table*}
\sisetup{
exponent-mode = fixed,
fixed-exponent = 0,
}
\begin{tiny}
\begin{center}
\begin{tabular}{l r r r r r r r r}
              &               &              & Baseline & SampLNS UB & SampLNS LB &  & SampLNS & \\
Feature Model & $|\features|$ & $|\clauses|$ & \emph{min}     & \emph{mean} (\emph{min}) & \emph{mean} (\emph{max}) & Savings & UB/LB & Time to Bounds\\
\hline
calculate & \num[text-series-to-math=true]{9} & \num[text-series-to-math=true]{15} & \num[text-series-to-math=true]{9} & \textbf{\num[text-series-to-math=true]{5}}  (\textbf{\num[text-series-to-math=true]{5}}) & \textbf{\num[text-series-to-math=true]{5}}  (\textbf{\num[text-series-to-math=true]{5}}) & \SI{44}{\percent} (\SI{44}{\percent})& 1.00 (1.00)& $<\SI{1}{\second}$ (\SI{1}{\second})\\
lcm & \num[text-series-to-math=true]{9} & \num[text-series-to-math=true]{16} & \num[text-series-to-math=true]{8} & \textbf{\num[text-series-to-math=true]{6}}  (\textbf{\num[text-series-to-math=true]{6}}) & \textbf{\num[text-series-to-math=true]{6}}  (\textbf{\num[text-series-to-math=true]{6}}) & \SI{25}{\percent} (\SI{25}{\percent})& 1.00 (1.00)& $<\SI{1}{\second}$ ($<\SI{1}{\second}$)\\
email & \num[text-series-to-math=true]{10} & \num[text-series-to-math=true]{17} & \textbf{\num[text-series-to-math=true]{6}} & \textbf{\num[text-series-to-math=true]{6}}  (\textbf{\num[text-series-to-math=true]{6}}) & \textbf{\num[text-series-to-math=true]{6}}  (\textbf{\num[text-series-to-math=true]{6}}) & \SI{0}{\percent} (\SI{0}{\percent})& 1.00 (1.00)& $<\SI{1}{\second}$ ($<\SI{1}{\second}$)\\
ChatClient & \num[text-series-to-math=true]{14} & \num[text-series-to-math=true]{20} & \textbf{\num[text-series-to-math=true]{7}} & \textbf{\num[text-series-to-math=true]{7}}  (\textbf{\num[text-series-to-math=true]{7}}) & \textbf{\num[text-series-to-math=true]{7}}  (\textbf{\num[text-series-to-math=true]{7}}) & \SI{0}{\percent} (\SI{0}{\percent})& 1.00 (1.00)& \SI{1}{\second} (\SI{2}{\second})\\
toybox\_2006-10-31\ldots & \num[text-series-to-math=true]{16} & \num[text-series-to-math=true]{13} & \num[text-series-to-math=true]{9} & \textbf{\num[text-series-to-math=true]{8}}  (\textbf{\num[text-series-to-math=true]{8}}) & \textbf{\num[text-series-to-math=true]{8}}  (\textbf{\num[text-series-to-math=true]{8}}) & \SI{11}{\percent} (\SI{11}{\percent})& 1.00 (1.00)& \SI{1}{\second} (\SI{1}{\second})\\
car & \num[text-series-to-math=true]{16} & \num[text-series-to-math=true]{33} & \num[text-series-to-math=true]{6} & \textbf{\num[text-series-to-math=true]{5}}  (\textbf{\num[text-series-to-math=true]{5}}) & \textbf{\num[text-series-to-math=true]{5}}  (\textbf{\num[text-series-to-math=true]{5}}) & \SI{17}{\percent} (\SI{17}{\percent})& 1.00 (1.00)& $<\SI{1}{\second}$ ($<\SI{1}{\second}$)\\
FeatureIDE & \num[text-series-to-math=true]{19} & \num[text-series-to-math=true]{27} & \num[text-series-to-math=true]{9} & \textbf{\num[text-series-to-math=true]{8}}  (\textbf{\num[text-series-to-math=true]{8}}) & \textbf{\num[text-series-to-math=true]{8}}  (\textbf{\num[text-series-to-math=true]{8}}) & \SI{11}{\percent} (\SI{11}{\percent})& 1.00 (1.00)& \SI{271}{\second} (\SI{128}{\second})\\
FameDB & \num[text-series-to-math=true]{22} & \num[text-series-to-math=true]{40} & \textbf{\num[text-series-to-math=true]{8}} & \textbf{\num[text-series-to-math=true]{8}}  (\textbf{\num[text-series-to-math=true]{8}}) & \textbf{\num[text-series-to-math=true]{8}}  (\textbf{\num[text-series-to-math=true]{8}}) & \SI{0}{\percent} (\SI{0}{\percent})& 1.00 (1.00)& \SI{1}{\second} (\SI{1}{\second})\\
APL & \num[text-series-to-math=true]{23} & \num[text-series-to-math=true]{35} & \num[text-series-to-math=true]{9} & \textbf{\num[text-series-to-math=true]{7}}  (\textbf{\num[text-series-to-math=true]{7}}) & \textbf{\num[text-series-to-math=true]{7}}  (\textbf{\num[text-series-to-math=true]{7}}) & \SI{22}{\percent} (\SI{22}{\percent})& 1.00 (1.00)& \SI{1}{\second} (\SI{1}{\second})\\
SafeBali & \num[text-series-to-math=true]{24} & \num[text-series-to-math=true]{45} & \textbf{\num[text-series-to-math=true]{11}} & \textbf{\num[text-series-to-math=true]{11}}  (\textbf{\num[text-series-to-math=true]{11}}) & \textbf{\num[text-series-to-math=true]{11}}  (\textbf{\num[text-series-to-math=true]{11}}) & \SI{0}{\percent} (\SI{0}{\percent})& 1.00 (1.00)& $<\SI{1}{\second}$ ($<\SI{1}{\second}$)\\
TightVNC & \num[text-series-to-math=true]{28} & \num[text-series-to-math=true]{39} & \num[text-series-to-math=true]{11} & \textbf{\num[text-series-to-math=true]{8}}  (\textbf{\num[text-series-to-math=true]{8}}) & \textbf{\num[text-series-to-math=true]{8}}  (\textbf{\num[text-series-to-math=true]{8}}) & \SI{27}{\percent} (\SI{27}{\percent})& 1.00 (1.00)& \SI{16}{\second} (\SI{21}{\second})\\
APL-Model & \num[text-series-to-math=true]{28} & \num[text-series-to-math=true]{40} & \num[text-series-to-math=true]{10} & \textbf{\num[text-series-to-math=true]{8}}  (\textbf{\num[text-series-to-math=true]{8}}) & \textbf{\num[text-series-to-math=true]{8}}  (\textbf{\num[text-series-to-math=true]{8}}) & \SI{20}{\percent} (\SI{20}{\percent})& 1.00 (1.00)& \SI{14}{\second} (\SI{15}{\second})\\
gpl & \num[text-series-to-math=true]{38} & \num[text-series-to-math=true]{99} & \num[text-series-to-math=true]{17} & \textbf{\num[text-series-to-math=true]{16}}  (\textbf{\num[text-series-to-math=true]{16}}) & \textbf{\num[text-series-to-math=true]{16}}  (\textbf{\num[text-series-to-math=true]{16}}) & \SI{5.9}{\percent} (\SI{5.9}{\percent})& 1.00 (1.00)& \SI{3}{\second} (\SI{3}{\second})\\
SortingLine & \num[text-series-to-math=true]{39} & \num[text-series-to-math=true]{77} & \num[text-series-to-math=true]{12} & \textbf{\num[text-series-to-math=true]{9}}  (\textbf{\num[text-series-to-math=true]{9}}) & \textbf{\num[text-series-to-math=true]{9}}  (\textbf{\num[text-series-to-math=true]{9}}) & \SI{25}{\percent} (\SI{25}{\percent})& 1.00 (1.00)& \SI{8}{\second} (\SI{9}{\second})\\
dell & \num[text-series-to-math=true]{46} & \num[text-series-to-math=true]{244} & \num[text-series-to-math=true]{32} & \textbf{\num[text-series-to-math=true]{31}}  (\textbf{\num[text-series-to-math=true]{31}}) & \textbf{\num[text-series-to-math=true]{31}}  (\textbf{\num[text-series-to-math=true]{31}}) & \SI{3.1}{\percent} (\SI{3.1}{\percent})& 1.00 (1.00)& \SI{29}{\second} (\SI{45}{\second})\\
PPU & \num[text-series-to-math=true]{52} & \num[text-series-to-math=true]{109} & \textbf{\num[text-series-to-math=true]{12}} & \textbf{\num[text-series-to-math=true]{12}}  (\textbf{\num[text-series-to-math=true]{12}}) & \textbf{\num[text-series-to-math=true]{12}}  (\textbf{\num[text-series-to-math=true]{12}}) & \SI{0}{\percent} (\SI{0}{\percent})& 1.00 (1.00)& \SI{2}{\second} (\SI{2}{\second})\\
berkeleyDB1 & \num[text-series-to-math=true]{76} & \num[text-series-to-math=true]{147} & \num[text-series-to-math=true]{19} & \textbf{\num[text-series-to-math=true]{15}}  (\textbf{\num[text-series-to-math=true]{15}}) & \textbf{\num[text-series-to-math=true]{15}}  (\textbf{\num[text-series-to-math=true]{15}}) & \SI{21}{\percent} (\SI{21}{\percent})& 1.00 (1.00)& \SI{77}{\second} (\SI{137}{\second})\\
axTLS & \num[text-series-to-math=true]{96} & \num[text-series-to-math=true]{183} & \num[text-series-to-math=true]{16} & \num[text-series-to-math=true]{11}  (\num[text-series-to-math=true]{11}) & \num[text-series-to-math=true]{10}  (\num[text-series-to-math=true]{10}) & \SI{31}{\percent} (\SI{31}{\percent})& 1.10 (1.10)& \SI{20}{\second} (\SI{20}{\second})\\
Violet & \num[text-series-to-math=true]{101} & \num[text-series-to-math=true]{203} & \num[text-series-to-math=true]{23} & \num[text-series-to-math=true]{17}  (\num[text-series-to-math=true]{17}) & \num[text-series-to-math=true]{16}  (\num[text-series-to-math=true]{16}) & \SI{26}{\percent} (\SI{26}{\percent})& 1.06 (1.06)& \SI{476}{\second} (\SI{656}{\second})\\
berkeleyDB2 & \num[text-series-to-math=true]{119} & \num[text-series-to-math=true]{346} & \num[text-series-to-math=true]{20} & \textbf{\num[text-series-to-math=true]{12}}  (\textbf{\num[text-series-to-math=true]{12}}) & \textbf{\num[text-series-to-math=true]{12}}  (\textbf{\num[text-series-to-math=true]{12}}) & \SI{40}{\percent} (\SI{40}{\percent})& 1.00 (1.00)& \SI{162}{\second} (\SI{282}{\second})\\
soletta\_2015-06-2\ldots & \num[text-series-to-math=true]{129} & \num[text-series-to-math=true]{192} & \num[text-series-to-math=true]{30} & \textbf{\num[text-series-to-math=true]{24}}  (\textbf{\num[text-series-to-math=true]{24}}) & \textbf{\num[text-series-to-math=true]{24}}  (\textbf{\num[text-series-to-math=true]{24}}) & \SI{20}{\percent} (\SI{20}{\percent})& 1.00 (1.00)& \SI{21}{\second} (\SI{60}{\second})\\
BattleofTanks & \num[text-series-to-math=true]{144} & \num[text-series-to-math=true]{769} & \num[text-series-to-math=true]{451} & \num[text-series-to-math=true]{320}  (\num[text-series-to-math=true]{295}) & \num[text-series-to-math=true]{256}  (\num[text-series-to-math=true]{256}) & \SI{29}{\percent} (\SI{35}{\percent})& 1.25 (1.15)& \SI{887}{\second} (\SI{160}{\minute})\\
BankingSoftware & \num[text-series-to-math=true]{176} & \num[text-series-to-math=true]{280} & \num[text-series-to-math=true]{40} & \textbf{\num[text-series-to-math=true]{29}}  (\textbf{\num[text-series-to-math=true]{29}}) & \textbf{\num[text-series-to-math=true]{29}}  (\textbf{\num[text-series-to-math=true]{29}}) & \SI{28}{\percent} (\SI{28}{\percent})& 1.00 (1.00)& \SI{306}{\second} (\SI{429}{\second})\\
fiasco\_2017-09-26\ldots & \num[text-series-to-math=true]{230} & \num[text-series-to-math=true]{1181} & \num[text-series-to-math=true]{234} & \num[text-series-to-math=true]{225}  (\textbf{\num[text-series-to-math=true]{225}}) & \textbf{\num[text-series-to-math=true]{225}}  (\textbf{\num[text-series-to-math=true]{225}}) & \SI{3.8}{\percent} (\SI{3.9}{\percent})& 1.00 (1.00)& \SI{382}{\second} (\SI{579}{\second})\\
fiasco\_2020-12-01\ldots & \num[text-series-to-math=true]{258} & \num[text-series-to-math=true]{1542} & \num[text-series-to-math=true]{209} & \num[text-series-to-math=true]{196}  (\textbf{\num[text-series-to-math=true]{196}}) & \textbf{\num[text-series-to-math=true]{196}}  (\textbf{\num[text-series-to-math=true]{196}}) & \SI{6.1}{\percent} (\SI{6.2}{\percent})& 1.00 (1.00)& \SI{438}{\second} (\SI{478}{\second})\\
uclibc\_2008-06-05\ldots & \num[text-series-to-math=true]{263} & \num[text-series-to-math=true]{1699} & \textbf{\num[text-series-to-math=true]{505}} & \textbf{\num[text-series-to-math=true]{505}}  (\textbf{\num[text-series-to-math=true]{505}}) & \textbf{\num[text-series-to-math=true]{505}}  (\textbf{\num[text-series-to-math=true]{505}}) & \SI{0}{\percent} (\SI{0}{\percent})& 1.00 (1.00)& \SI{104}{\second} (\SI{67}{\second})\\
uclibc\_2020-12-24\ldots & \num[text-series-to-math=true]{272} & \num[text-series-to-math=true]{1670} & \textbf{\num[text-series-to-math=true]{365}} & \textbf{\num[text-series-to-math=true]{365}}  (\textbf{\num[text-series-to-math=true]{365}}) & \textbf{\num[text-series-to-math=true]{365}}  (\textbf{\num[text-series-to-math=true]{365}}) & \SI{0}{\percent} (\SI{0}{\percent})& 1.00 (1.00)& \SI{108}{\second} (\SI{112}{\second})\\
E-Shop & \num[text-series-to-math=true]{326} & \num[text-series-to-math=true]{499} & \num[text-series-to-math=true]{19} & \num[text-series-to-math=true]{12}  (\num[text-series-to-math=true]{12}) & \num[text-series-to-math=true]{9}  (\num[text-series-to-math=true]{10}) & \SI{37}{\percent} (\SI{37}{\percent})& 1.30 (1.20)& \SI{268}{\second} (\SI{64}{\minute})\\
toybox\_2020-12-06\ldots & \num[text-series-to-math=true]{334} & \num[text-series-to-math=true]{92} & \num[text-series-to-math=true]{18} & \num[text-series-to-math=true]{13}  (\num[text-series-to-math=true]{13}) & \num[text-series-to-math=true]{7}  (\num[text-series-to-math=true]{8}) & \SI{28}{\percent} (\SI{28}{\percent})& 1.71 (1.62)& \SI{532}{\second} (\SI{35}{\minute})\\
DMIE & \num[text-series-to-math=true]{366} & \num[text-series-to-math=true]{627} & \num[text-series-to-math=true]{26} & \textbf{\num[text-series-to-math=true]{16}}  (\textbf{\num[text-series-to-math=true]{16}}) & \textbf{\num[text-series-to-math=true]{16}}  (\textbf{\num[text-series-to-math=true]{16}}) & \SI{38}{\percent} (\SI{38}{\percent})& 1.00 (1.00)& \SI{104}{\second} (\SI{135}{\second})\\
soletta\_2017-03-0\ldots & \num[text-series-to-math=true]{458} & \num[text-series-to-math=true]{1862} & \num[text-series-to-math=true]{56} & \textbf{\num[text-series-to-math=true]{37}}  (\textbf{\num[text-series-to-math=true]{37}}) & \num[text-series-to-math=true]{31}  (\textbf{\num[text-series-to-math=true]{37}}) & \SI{34}{\percent} (\SI{34}{\percent})& 1.16 (1.00)& \SI{387}{\second} (\SI{24}{\minute})\\
busybox\_2007-01-2\ldots & \num[text-series-to-math=true]{540} & \num[text-series-to-math=true]{429} & \num[text-series-to-math=true]{34} & \textbf{\num[text-series-to-math=true]{21}}  (\textbf{\num[text-series-to-math=true]{21}}) & \textbf{\num[text-series-to-math=true]{21}}  (\textbf{\num[text-series-to-math=true]{21}}) & \SI{38}{\percent} (\SI{38}{\percent})& 1.00 (1.00)& \SI{164}{\second} (\SI{237}{\second})\\
fs\_2017-05-22 & \num[text-series-to-math=true]{557} & \num[text-series-to-math=true]{4992} & \num[text-series-to-math=true]{398} & \textbf{\num[text-series-to-math=true]{396}}  (\textbf{\num[text-series-to-math=true]{396}}) & \textbf{\num[text-series-to-math=true]{396}}  (\textbf{\num[text-series-to-math=true]{396}}) & \SI{0.5}{\percent} (\SI{0.5}{\percent})& 1.00 (1.00)& \SI{478}{\second} (\SI{575}{\second})\\
WaterlooGenerated & \num[text-series-to-math=true]{580} & \num[text-series-to-math=true]{879} & \num[text-series-to-math=true]{144} & \textbf{\num[text-series-to-math=true]{82}}  (\textbf{\num[text-series-to-math=true]{82}}) & \textbf{\num[text-series-to-math=true]{82}}  (\textbf{\num[text-series-to-math=true]{82}}) & \SI{43}{\percent} (\SI{43}{\percent})& 1.00 (1.00)& \SI{223}{\second} (\SI{310}{\second})\\
financial\_services & \num[text-series-to-math=true]{771} & \num[text-series-to-math=true]{7238} & \num[text-series-to-math=true]{4384} & \num[text-series-to-math=true]{4368}  (\num[text-series-to-math=true]{4340}) & \num[text-series-to-math=true]{4274}  (\num[text-series-to-math=true]{4336}) & \SI{0.36}{\percent} (\SI{1}{\percent})& 1.02 (1.00)& \SI{862}{\second} (\SI{102}{\minute})\\
busybox-1\_18\_0 & \num[text-series-to-math=true]{854} & \num[text-series-to-math=true]{1164} & \num[text-series-to-math=true]{26} & \num[text-series-to-math=true]{16}  (\num[text-series-to-math=true]{16}) & \num[text-series-to-math=true]{11}  (\num[text-series-to-math=true]{13}) & \SI{35}{\percent} (\SI{38}{\percent})& 1.53 (1.23)& \SI{233}{\second} (\SI{59}{\minute})\\
busybox-1\_29\_2 & \num[text-series-to-math=true]{1018} & \num[text-series-to-math=true]{997} & \num[text-series-to-math=true]{36} & \num[text-series-to-math=true]{22}  (\num[text-series-to-math=true]{22}) & \num[text-series-to-math=true]{17}  (\num[text-series-to-math=true]{21}) & \SI{38}{\percent} (\SI{39}{\percent})& 1.26 (1.05)& \SI{465}{\second} (\SI{60}{\minute})\\
busybox\_2020-12-1\ldots & \num[text-series-to-math=true]{1050} & \num[text-series-to-math=true]{996} & \num[text-series-to-math=true]{33} & \num[text-series-to-math=true]{21}  (\num[text-series-to-math=true]{20}) & \num[text-series-to-math=true]{17}  (\num[text-series-to-math=true]{19}) & \SI{36}{\percent} (\SI{39}{\percent})& 1.19 (1.05)& \SI{407}{\second} (\SI{17}{\minute})\\
am31\_sim & \num[text-series-to-math=true]{1178} & \num[text-series-to-math=true]{2747} & \num[text-series-to-math=true]{60} & \num[text-series-to-math=true]{36}  (\num[text-series-to-math=true]{33}) & \num[text-series-to-math=true]{26}  (\num[text-series-to-math=true]{29}) & \SI{39}{\percent} (\SI{45}{\percent})& 1.36 (1.14)& \SI{699}{\second} (\SI{77}{\minute})\\
EMBToolkit & \num[text-series-to-math=true]{1179} & \num[text-series-to-math=true]{5414} & \num[text-series-to-math=true]{1881} & \num[text-series-to-math=true]{1879}  (\textbf{\num[text-series-to-math=true]{1872}}) & \num[text-series-to-math=true]{1821}  (\textbf{\num[text-series-to-math=true]{1872}}) & \SI{0.1}{\percent} (\SI{0.48}{\percent})& 1.03 (1.00)& \SI{863}{\second} (\SI{47}{\minute})\\
atlas\_mips32\_4kc & \num[text-series-to-math=true]{1229} & \num[text-series-to-math=true]{2875} & \num[text-series-to-math=true]{66} & \num[text-series-to-math=true]{38}  (\num[text-series-to-math=true]{36}) & \num[text-series-to-math=true]{31}  (\num[text-series-to-math=true]{33}) & \SI{41}{\percent} (\SI{45}{\percent})& 1.22 (1.09)& \SI{548}{\second} (\SI{50}{\minute})\\
eCos-3-0\_i386pc & \num[text-series-to-math=true]{1245} & \num[text-series-to-math=true]{3723} & \num[text-series-to-math=true]{64} & \num[text-series-to-math=true]{43}  (\num[text-series-to-math=true]{39}) & \num[text-series-to-math=true]{31}  (\num[text-series-to-math=true]{36}) & \SI{32}{\percent} (\SI{39}{\percent})& 1.38 (1.08)& \SI{621}{\second} (\SI{146}{\minute})\\
integrator\_arm7 & \num[text-series-to-math=true]{1272} & \num[text-series-to-math=true]{2980} & \num[text-series-to-math=true]{66} & \num[text-series-to-math=true]{38}  (\num[text-series-to-math=true]{36}) & \num[text-series-to-math=true]{30}  (\num[text-series-to-math=true]{33}) & \SI{41}{\percent} (\SI{45}{\percent})& 1.28 (1.09)& \SI{681}{\second} (\SI{82}{\minute})\\
XSEngine & \num[text-series-to-math=true]{1273} & \num[text-series-to-math=true]{2942} & \num[text-series-to-math=true]{63} & \num[text-series-to-math=true]{38}  (\num[text-series-to-math=true]{36}) & \num[text-series-to-math=true]{31}  (\num[text-series-to-math=true]{32}) & \SI{39}{\percent} (\SI{43}{\percent})& 1.23 (1.12)& \SI{572}{\second} (\SI{52}{\minute})\\
aaed2000 & \num[text-series-to-math=true]{1298} & \num[text-series-to-math=true]{3036} & \num[text-series-to-math=true]{87} & \num[text-series-to-math=true]{55}  (\num[text-series-to-math=true]{52}) & \num[text-series-to-math=true]{51}  (\num[text-series-to-math=true]{51}) & \SI{36}{\percent} (\SI{40}{\percent})& 1.09 (1.02)& \SI{707}{\second} (\SI{75}{\minute})\\
FreeBSD-8\_0\_0 & \num[text-series-to-math=true]{1397} & \num[text-series-to-math=true]{15692} & \num[text-series-to-math=true]{76} & \num[text-series-to-math=true]{47}  (\num[text-series-to-math=true]{41}) & \num[text-series-to-math=true]{27}  (\num[text-series-to-math=true]{30}) & \SI{38}{\percent} (\SI{46}{\percent})& 1.72 (1.37)& \SI{831}{\second} (\SI{120}{\minute})\\
ea2468 & \num[text-series-to-math=true]{1408} & \num[text-series-to-math=true]{3319} & \num[text-series-to-math=true]{65} & \num[text-series-to-math=true]{38}  (\num[text-series-to-math=true]{36}) & \num[text-series-to-math=true]{31}  (\num[text-series-to-math=true]{32}) & \SI{41}{\percent} (\SI{45}{\percent})& 1.24 (1.12)& \SI{721}{\second} (\SI{67}{\minute})\\

\hline
\ \\
\textbf{optimality} & & & 7 & $\geq 26$ & & \\ & & & [\SI{15}{\percent}] &[\SI{55}{\percent}] & & \\ 
\textbf{improvements} & & & & &\multicolumn{2}{r}{\num[text-series-to-math=true]{40} [\SI{85}{\percent}]} \\ \ \\ \ \\

\end{tabular}
\end{center}
\end{tiny}
\caption{
Comparison of SampLNS with algorithms from the literature. Here, ``Baseline'' indicates the best sample size found by any of the existing algorithms, each having five runs of $\SI{15}{\minute}$.
``SampLNS UB'' shows the sample size found by SampLNS with the initial numbers  indicating the mean over five runs of \SI{15}{\minute} each, i.e., typical outcomes for a short run, while the numbers in parentheses describe the best of five extended runs of up to \SI{3}{\hour}, i.e., the potential outcomes of longer runs.
Similarly, ``SampLNS LB'' is the lower bound, while ``Savings'' quantifies the reduction in sample size achieved by SampLNS compared to previous algorithms.
``SampLNS UB/LB'' illustrates the relationship between upper and lower bounds, with a value of 1.0 implying optimality. 
``Time to bounds''  records the moment when final lower and upper bounds were established, potentially preceding the time limit.
A low value may signal efficiency, yet it could also suggest stagnation, depending on the solution quality at that moment.
}\label{tab:eval:new_benchmark_results}
\end{table*}

The size of the smallest sample found by any existing algorithm for each feature model is listed in the column \emph{Baseline} of \cref{tab:eval:new_benchmark_results}.
The best lower bound found by SampLNS is listed in parentheses in the column \emph{SampLNS LB}, and the column \emph{SampLNS UB} shows in parentheses the sample size of the smallest known sample (all achieved by SampLNS), defining the range in which the smallest possible sample can lie.
For \num{7} of the \num{47} feature models, the smallest sample found by any existing algorithm is equal to the lower bound.
For at least \num{28} feature models, the difference to the optimum is above \SI{20}{\percent}.
\Cref{fig:eval:diff_to_lb} displays the differences to the lower bounds for each algorithm separately.
A value of \SI{20}{\percent} would imply that the sample size is $1.2\times$ the lower bound.
We can see that of the previous algorithms, YASA with $m=10$ (resp. the \SI{15}{\minute}-variant) achieves the smallest samples, with a median of \SI{46}{\percent} (resp. \SI{42}{\percent}) above the lower bound.
The other algorithms yield significantly larger samples, with Incling yielding the largest samples, with a median of more than twice the size of the lower bound.
As not all feature models could be solved by all algorithms, the number of feature models solved by each algorithm is added in square brackets.
Only the YASA variants and SampLNS could compute full samples for all feature models in time.
Incling and ICPL failed for one feature model, namely \texttt{financial\_serivces}, to compute a feasible sample within the time limit.
The other algorithms failed on multiple feature models.

\subsubsection{(RQ2) \rqtwo} 

The mean sample sizes over five runs of SampLNS (with a \SI{900}{\second} time limit) are listed in the column \emph{SampLNS UB} of \cref{tab:eval:new_benchmark_results}.
The values in parentheses are the smallest sample size found by SampLNS when running for up to \SI{3}{\hour}, and only indicate the potential of SampLNS to further reduce the sample size.
For \num{40} (\SI{85}{\percent}) of the models, the \emph{mean} sample size of SampLNS (after \SI{900}{\second}) was smaller than the \emph{best} samples of any other algorithm.
For \num{33} (\SI{70}{\percent}) of the models, the mean sample size of SampLNS was at least \SI{10}{\percent} smaller than the smallest sample found by any other algorithm.
At least \SI{55}{\percent} of the samples by SampLNS are provably optimal.
On average, the samples of SampLNS are \SI{22}{\percent} smaller than the smallest sample found by any existing algorithm.
The first bar in \cref{fig:eval:diff_to_lb} shows that the \SI{75}{\percent} quantile is with \SI{10}{\percent} smaller than the \SI{25}{\percent} quantile of YASA ($m=10$) with \SI{11}{\percent}.
To evaluate the variance of SampLNS, we computed the standard deviation of the relative differences from the best value, expressed as percentages.
The mean standard deviation is \SI{0.47}{\percent} for the upper bound and \SI{2.62}{\percent} for the lower bound.

\subsubsection{(RQ3) \rqthree}
\begin{figure}[tb]
    \centering
    \includegraphics[width=0.6\columnwidth]{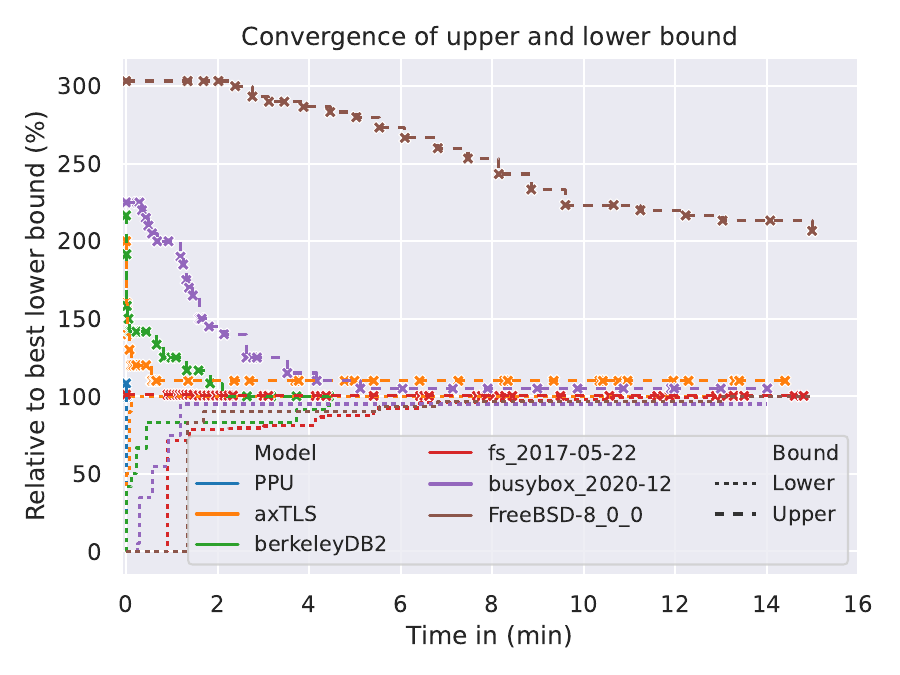}
    \caption{
    Convergence of lower and upper bound over time for a selection of SampLNS runs.
    Values are relative to the best lower bound of the model,
    so lower and upper bound meeting at \SI{100}{\percent} indicates provable optimality. 
    Lower and upper bound usually make quick progress already in the first seconds and minutes.
    The crosses on a curve indicate the end of an upper bound optimization step.
    The lower bound is only queried at the same time, such that not all increments are shown.
    }
    \label{fig:eval:convergence_lb_ub_selection}
\end{figure}
The average times achieved by SampLNS until the last improvement on the lower or upper bound are listed in column \emph{Time to bounds} in \cref{tab:eval:new_benchmark_results}.
At this time we could have aborted SampLNS while still achieving the same result.
If SampLNS achieved a provably optimal sample size, it aborted at this time. 
For \num{21} of the feature models, this time is below \SI{3}{\minute}.
\Cref{fig:eval:convergence_lb_ub_selection} shows the convergence of lower and upper bound over time for a selection of models.
The curves for the remaining models are attached as artifact.
 \emph{PPU} starts with an already optimal sample and terminates quickly after the matching lower bound has been found.
 The sample of \emph{axTLS} gets quickly optimized but no matching lower bound can be found.
 For \emph{berkeleyDB2}, we achieve strong improvements in the first seconds and then make continuous progress for two minutes.
 Optimality is proved in around four minutes.
 \emph{fs\_2017-05-22} already starts with a nearly optimal sample, which we only improve slightly and it takes around seven minutes until we have a nearly tight lower bound.
 We also see multiple short optimization steps interleaved by a few longer steps.
 \emph{busybox\_2020-12-16} improves continuously over the first five minutes, first with many quick steps that become longer at the end.
 \emph{FreeBSD-8\_0\_0} needs a few optimization steps before it makes any progress, but then makes continuous progress until the timeout.

\subsubsection{(RQ4) \rqfour}
\begin{figure}[htb]
    \centering
    \includegraphics[width=0.6\columnwidth]{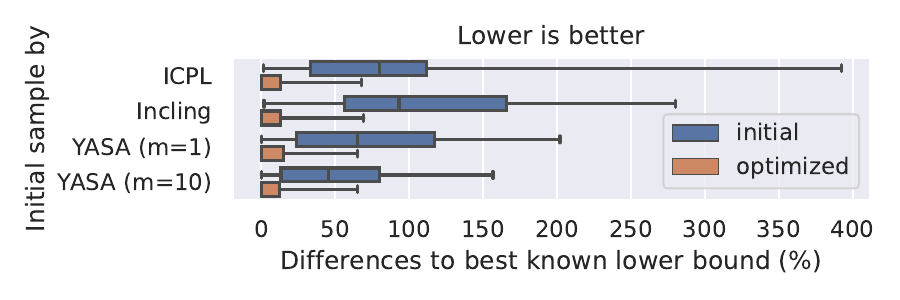}
    \caption{
        Analogous to \cref{fig:eval:diff_to_lb}, but comparing different versions of SampLNS where the initial sample was produced by different algorithms, to compare the dependence of SampLNS on the size of the initial sample.
        ICPL and Incling have significantly larger initial samples than YASA ($m=10$) as can be seen by the blue boxes, but SampLNS can still reduce them to a similar size as can be seen by the orange boxes.
        For YASA ($m=1$) and YASA ($m=10$), the optimized samples are nearly identical, despite a visible difference in the initial samples.
        This indicates that the size of the initial sample is of low importance for SampLNS.
    }\label{fig:initial_sample_comparison}
\end{figure}
\Cref{fig:initial_sample_comparison} shows the variance of the sample sizes of SampLNS when exchanging the initial sample with samples generated by other algorithms.
The runtime for computing the initial sample is not included in the time limit of \SI{15}{\minute}, to allow for a fair comparison.
Only the \num{46} feature models for which all algorithms could compute an initial sample in time are included.
The sample sizes of the four selected algorithms have shown a high variance in \cref{fig:eval:diff_to_lb}, but after optimization with SampLNS they all have a \SI{75}{\percent} quantile of \SIrange{12}{15}{\percent} above the lower bound.
The sample sizes of SampLNS are nearly identical when using YASA~($m=1$) and YASA~($m=10$), despite a visible difference in the initial samples.
While the other two algorithms result in slightly larger samples, the difference is small when compared to the differences between the initial samples of the algorithms themselves.

\subsection{Discussion}

\subsubsection{(RQ1) \rqone}
Yes, but only for few feature models in our benchmark set.
For most feature models, none of the algorithms could compute a sample within \SI{20}{\percent} of neither the lower bound nor the smallest sample found using SampLNS.
The best performing algorithm in our benchmark, YASA with $m=10$, achieves for only \SI{25}{\percent} of the samples a gap of less than \SI{13}{\percent} (cf. \cref{fig:eval:diff_to_lb}).
As the effect of larger $m$ for YASA is quickly diminishing, it also cannot be expected that fundamental improvements can be achieved by just increasing~$m$.
This is demonstrated by comparing YASA with \( m=10 \), which takes an average of \SI{15}{\second}, to the variant that utilizes the full \SI{900}{\second} yet achieves only marginal improvements.

\subsubsection{(RQ2) \rqtwo}
Yes, and often substantially.
SampLNS could minimize the samples for \SI{85}{\percent} of models in our benchmark set, reducing \SI{55}{\percent} of them to minimality and \SI{82}{\percent} below $1.2\times$ the lower bound.
This not only highlights the performance of SampLNS, but also the quality of the lower bound for many feature models.
It is worth mentioning that even small improvements on sample size are relevant for practical applications, as testing effort in terms of time and cost can be considerable for each tested configuration~\cite{HNA+:EMSE19}.
Especially in industrial product lines, consisting not only of software but also hardware components, for instance in the automotive domain, the effort and cost for assembling and testing even a single product can be enormous~\cite{HPS+:GPCE22}.

However, there are some feature models where the average sample size is still more than $\num{1.7}\times$ the lower bound.
The sample size for the \emph{FreeBDS-8\_0\_0} feature model was on average \num{47}, whereas the best lower bound is only \num{30}, leaving a large gap.
Determining whether this gap stems from limitations of the lower bound or the sample generation is challenging, given the absence of knowledge regarding the true optimal sample size.
The extended runtime experiments suggest there is room for enhancement on both fronts for models that have not been solved to optimality within \SI{15}{\minute}.

The standard deviation of the lower bound, at \SI{2.62}{\percent}, is noticeably higher than that of the sample size, which is \SI{0.47}{\percent}.
However, both values are relatively low, indicating that SampLNS consistently delivers reliable and good results.

\subsubsection{(RQ3) \rqthree}

SampLNS is efficient enough to optimize samples to (near-)optimality and compute lower bounds in seconds for small and in minutes for many medium-sized feature models.
On most models, SampLNS makes visible improvements on the first few optimization steps, indicating a usefulness even for short optimization times.
However, for some models the lower and upper bounds could still be visibly improved when running SampLNS for up to \SI{3}{\hour} instead of only \SI{900}{\second}.
The runtime of SampLNS is of course significantly higher than YASA~($m=1$), but while increasing the optimization parameter $m$ of YASA strongly increases the runtime and only leads to small improvements, SampLNS can achieve much larger improvements with a reasonable increase in runtime.

\subsubsection{(RQ4) \rqfour}

No, the quality of the initial sample has only a negligible impact on SampLNS.
Even when using the algorithm with the largest initial sample, SampLNS can still reduce the sample size to a similar size as when using the algorithm with the smallest initial sample.
In particular, the increased optimization of YASA ($m=10$) does not lead to a visible difference in the optimized sample size.
This indicates that the quality of the initial sample is of low importance for SampLNS, and YASA could indeed be replaced by other algorithms without a significant impact on the results.
As YASA ($m=1$) is with an average of \SI{2.1}{\second} the fastest algorithm for generating an initial sample, it remains the best choice for SampLNS.

\subsection{Threats to Validity}

\subsubsection*{Internal Validity}

With SampLNS, we reduce the sample size of a given sample, which brings the risk of lowering the $t$-wise coverage and introducing invalid configurations.
To verify that the optimized samples are indeed correct (excluding errors in parser), we check the interactions of the initial samples for equality with the interactions of the optimized sample, and each configuration for feasibility.
Finally, the coverage of all samples of all algorithms are checked for equality to also ensure correctness of the baseline algorithms.

In our evaluation, we use multiple sampling algorithms, as baseline and as input for SampLNS.
Some algorithms are sensitive to the feature order of the input feature model, which leads to different samples.
To mitigate the risk of this potential bias, we used a large corpus of different feature models and, in addition, repeated the sampling process for each feature model with each algorithm five times and selected the mean of the results.
Furthermore, some sampling algorithms allow for additional parametrization.
In particular, YASA allows to set its parameter $m$ for the number of resamplings, which has an effect on the sample sizes.
To mitigate this threat, we run YASA with different values for $m$, which are within a range that was also used in other measurements~\cite{KTS+:VaMoS20}.

Repeated runs of SampLNS on the same feature model can lead to vastly different samples and certificates for the lower bound due to randomization.
To mitigate this threat, we run SampLNS five times for each feature model and use the mean of the results.
While the quality may differ, the randomization has no effect on the correctness of the results, as each optimization step
preserves a valid t-wise sample.

We are comparing the sample sizes relative to the best lower bound found by SampLNS in \cref{fig:eval:diff_to_lb,fig:eval:convergence_lb_ub_selection,fig:initial_sample_comparison}.
However, the lower bound may not guarantee the absolute smallest sample size; it only ensures it is not larger.
If we had a more precise lower bound, it could highlight differences more, showing our current approach tends to be cautious, possibly even disadvantaging SampLNS.

\subsubsection*{External Validity}

In our evaluation, we used a wide range of different feature models \cf{\autoref{subsec:experiment_design}}.
In this study, we took care to include feature models from many relevant domains, but did not include some large-scale systems, such as the Linux kernel.
At this time, this leaves generalizing our results for feature models of other domains and with 
even higher numbers of features for followup studies; see~\cref{sec:outlook} for an outlook.

We run SampLNS with specific parameter settings \cf{\autoref{sec:implementation}}.
Thus, we cannot generalize our results to other settings.
However, as we did not focus on optimizing our initial settings, we can assume that an optimized parameter set \eg{by using hyper parameter tuning} can yield even better results.

All sampling algorithms that we used in the evaluation are based on greedy strategies.
While there are evolutionary algorithms that also guarantee a full $t$-wise coverage~\cite{GCD:ESE11,MGSH:SPLC13},
previous work suggests that such algorithms can be expected to produce results similar to the algorithms we considered~\cite{AKT+:GPCE16}.

\section{Related Work}
In the following, we state our work's relation to works on finding upper and lower bounds for $t$-wise samples, using other sample quality metrics and sampling strategies, applying general methods on test effort reduction for configurable systems, and computing covering arrays \ie{samples} for unconstrained configurable systems.

\subsubsection*{Lower Bounds for $t$-Wise Samples}

Yamada et al.~\cite{yamada2015optimization} introduced CALOT for combinatorial testing, which uses an incremental SAT formulation that narrows down the sample size from an upper limit until it becomes infeasible, ultimately yielding a minimal sample.
CALOT shares similarities with the CP-SAT formulation discussed in \cref{sec:upper-bounds:cpsat}, but it employs a less powerful symmetry-breaking technique that would only yield weak lower bounds on binary features.
Unlike our LNS approach, which partitions the solution space to maintain a manageable CP-SAT formulation size, CALOT's SAT formulation can become unwieldy.
Another drawback is that its guarantee of optimality relies on the correctness of its implementation, and it cannot provide quality assurances if it runs out of time.
While SAT-solvers can provide proofs for infeasibility, those proofs are significantly more complex than our certificates.
Building on CALOT, Ans{\'o}tegui et al.~\cite{ansotegui_et_al:LIPIcs.CP.2021.12,ansotegui2022incomplete} enhanced algorithm scalability, optimizing instances with even $t=3,4,5$, but incomplete coverage.
While using a refinement operator, the approaches are not as dynamic as our LNS.
Furthermore, their lower bounding still just relies on a subset of $t$ combinatorial features, unlike our comprehensive computation of maximal mutually excluding sets, potentially spanning all available features, which is essential for models only containing binary features.

\subsubsection*{Upper Bound for $t$-Wise Samples}

There are numerous sampling algorithms for configurable systems that are designed for different use cases, based on different sampling strategies~\cite{VAT+:SPLC18,AZAB17,LFRE:IWCT15}.
Most $t$-wise sampling algorithms, such as YASA~\cite{KTS+:VaMoS20}, IncLing~\cite{AKT+:GPCE16}, UnWise~\cite{HSO+:VAMOS24}, ScalableCA~\cite{luo2024beyond}, 
ICPL~\cite{MHF:SPLC12}, IPOG~\cite{LK+:ECBS07}, Chv\'{a}tal's algorithm~\cite{JHF:MODELS11}, CASA~\cite{GCD:ESE11}, Alloy-based sampling~\cite{PSK+:ICST10}, MoSo-PoLiTe~\cite{OMR:SPLC10}, CALOT~\cite{yamada2015optimization}, and AETG-SAT~\cite{CDS:TSE08} focus on finding small samples for full $t$-wise coverage.
Consequently, any sample provided by one of these algorithms determines an upper bound for the used feature model; this is also the case for a recently proposed local search heuristic by Zhao et al.~\cite{zhao2023campactor} on a different set of benchmarks.
However, the distance from an optimal size of a $t$-wise sample (certified by a tight lower bound) was only known for very small instance sizes.   This differs from our work, which focuses on computing \emph{both}, upper \emph{and} lower bounds that are close or equal to the optimal size of a pairwise sample.

\subsubsection*{Sample Quality Metrics}

SampLNS, tries to find optimal $t$-wise samples in terms of sample size.
However, the quality of sampling algorithms can also be assessed by other criteria~\cite{VAT+:SPLC18}, including sampling efficiency~\cite{PKT+:SPLC24,KTS+:VaMoS20,AKT+:GPCE16,MHF:SPLC12,GCD:ESE11,MGSH:SPLC13,HLHE:VaMoS13,OZML:VaMoS11}, effectiveness~\cite{HNA+:EMSE19,MKR+:ICSE16,SRK+:SQJ12,AKT+:GPCE16}, and stability~\cite{PKR+:VaMoS21}.
Regarding sampling efficiency, SampLNS allows to specify a time limit for optimizing a sample, similar to many evolutionary algorithms~\cite{HPP+:TSE14,CR:IJAST14,HPL:SBSE14,EBG:CAiSE12,FLV:SBES17,HPP+:SPLC13}.
%
Determining the effectiveness of a sample is dependent on the particular use case, such as fault detection potential for testing~\cite{HNA+:EMSE19,MKR+:ICSE16}.
However, as this is hard to evaluate, many other coverage criteria are used as a proxy instead~\cite{VAT+:SPLC18}.

\subsubsection*{Beyond $t$-Wise Sampling}

SampLNS considers $t$-wise interaction coverage, but there exist many other coverage criteria for which samples can be generated~\cite{VAT+:SPLC18}, such as partial $t$-wise coverage~\cite{MKR+:ICSE16,BLK:ESECFSE20,KSS:VariComp13,JHF+:MODELS12},
coverage of the solution space~\cite{HPS+:GPCE22,KTS+:TR22,SCD:FASE12,TLD:OSR12},
coverage of test cases~\cite{HSMI:ICST14,KMKB:ESECFSE13,KBK:AOSD11},
uniform sampling~\cite{HFGB:SPLC20,PAP+:ICST19,AZT:SAT18,SGRM:LPAIR18,OGB:SPLC19},
distance-based sampling~\cite{KGS+:ICSE19,HPP+:TSE14}, and
coverage of feature model mutations~\cite{RBR+:SPLC15,AGV:ICST15,HPL:SBSE14}.
As some of these coverage criteria are correlated to $t$-wise coverage, it is reasonable to assume that LNS may be employed to reduce the size of corresponding samples as well.

\subsubsection*{Test Effort Reduction for Configurable Systems}

The ultimate goal of reducing sample sizes, in context of testing configurable systems, is to reduce the test effort.
Besides sample size reduction, this goal may also be achieved by, for instance, test case selection~\cite{BL:FOSD14,HSMI:ICST14}, test case prioritization~\cite{LATS:VaMoS17,ALL+:VACE17,ATL+:SoSyM19}, and regression testing~\cite{YH:STVR12,LMTS:ICSR16,LNTS:JSS19,NNT+:SEN19}.
Such approaches are orthogonal to our technique and may be used complementary to SampLNS.

\subsubsection*{Analysis Strategies for Configurable Systems}

Using sample-based testing, such as t-wise sampling, is just one strategy for analyzing interactions in configurable systems.
Alternative approaches encompass family-based~\cite{WMK:ESECFSE18,WMLK:OOPSLA18,TAK+:CSUR14,NNT+:SEN19} and regression-based analysis~\cite{YH:STVR12,LMTS:ICSR16,LNTS:JSS19,NNT+:SEN19}.
In contrast to sample-based testing, which can rely on single-system tooling, family- and regression-based analyses require special tooling to work~\cite{TAK+:CSUR14}.
In addition, family-based testing often does not scale for large-scale systems~\cite{KRE+:FOSD12,MWK+:ASE16}.

\subsubsection*{Covering Arrays for Unconstrained Configurable Systems}
Finding lower and upper bounds for unconstrained combinatorial interaction testing has received intensive attention~\cite{CDFP:TSE97,Hartman2005,CKO:DM12,KS:DM73,Katona1973}.
For unconstrained configuration spaces, Cohen et al.~\cite{CDFP:TSE97} show that the minimal size for a $t$-wise sample grows logarithmically in the number of features.
Furthermore, for unconstrained Boolean feature models, optimal pairwise samples are known for any number of features~\cite{Hartman2005,KS:DM73,Katona1973} and reasonable lower and upper bounds are known for $t$-wise samples~\cite{Hartman2005,CKO:DM12}.
We consider constrained Boolean feature models, which makes finding an optimal $t$-wise sample and even lower and upper bounds substantially more difficult~\cite{WNP+:arxiv19}.

\section{Conclusion}
\label{sec:outlook}

Quality assurance of modern software systems is often challenged by their large configuration spaces. 
The sample size is of crucial importance, as it is generally considered the dominating factor for the test effort for configurable systems~\cite{CDS:TSE08,Hartman2005,KS:DM73,Katona1973}.
In the last two decades, many heuristics have been proposed to derive sample configurations, but the quality of their results could only be compared to other heuristics so far. 
In particular, it has been unclear \emph{how low we can go} when minimizing sample size of $t$-wise interaction samples,
i.e., how close heuristically achieved solutions actually are to being optimal.

With SampLNS, we propose a new algorithm that offers a significant advancement in minimizing the size of $t$-wise interaction samples to near-optimality.
Importantly, it also provides easily verifiable lower bounds to prove the solution quality.
Our algorithm relies on Constraint Programming and Mixed Integer Programming,
which are scaled to the complexity of configuration spaces by means of a large
neighborhood search.  Our empirical evaluation shows that SampLNS scales
to configuration spaces with at least up to 1400 configuration options. 
For \num{26} of \num{47} subject systems, sample sizes of SampLNS match the lower bounds,
so we can provide certificates that prove their optimality.
In contrast, the samples computed by any previous algorithm are optimal for only seven systems. 
Even if samples generated with SampLNS are not guaranteed to be optimal for \num{21}
systems, sample size and lower bounds are often close (e.g., off by only one for five
systems).  In addition, SampLNS was able to improve the sample size for \num{40} of
\num{47} systems compared to the previous best algorithm for each system.

Our algorithm significantly reduces the testing effort for small- and medium-sized configurable systems.
Additionally, our lower bounds offer a clear framework for potential improvements, enabling the differentiation between instances with a small remaining optimality gap (and thus minimal room for improvement) and those with a large gap (and consequently greater potential for improvement).
These insights open up various directions for future work.

One promising avenue is extending the range of instance sizes that can be solved to provable optimality;
currently, the largest such instance is \texttt{EMBToolkit} with \num{1179} features.
Achieving this goal may be possible through a combination of hyper-parameter tuning and more efficient data structures.
While we do not expect this approach to yield provably optimal solutions for very large instances of practical importance 
--- such as those arising from the Linux kernel, which features nearly \num{20000} options ---
we remain optimistic that a spectrum of Algorithm Engineering methods could provide substantial improvements.
These may include further refinements to achieve better upper bounds and mathematical enhancements to the dual problem, aimed at generating provable lower bounds when the current dual leaves a non-trivial duality gap.

Moreover, we are hopeful that the fundamental ideas demonstrated in this work can be applied to other computationally complex problems in configurable systems beyond minimizing the sample size in $t$-wise interaction sampling.

\paragraph*{Acknowledgments}
This work was partially supported through the projects
``Quantum methods and Benchmarks for Ressource Allocation'' (QuBRA), German Federal Ministry for Education and Science (BMBF), grant number 13N16053  and ``Quantum Readiness for Optimization Providers'' (ProvideQ), German Federal Ministry of Economic Affairs and Climate Action (BMWK), grant number 01MQ22006A.

All authors contributed significantly to the work. The order of the first three authors is alphabetical.

\bibliographystyle{unsrt}
\bibliography{MYabrv,literature,main}

\end{document}